\documentclass[11pt]{amsart}

\usepackage{macros,amsaddr,physics}
\numberwithin{equation}{section}

\renewcommand{\div}{\del_{\Omega}}

\renewcommand{\op}{\operatorname}

\begin{document}

\title{A complex geometric perspective on the $a,c$ anomalies}
\author{Brian R. Williams}
\thanks{Boston University, Department of Mathematics and Statistics}
\email{bwill22@bu.edu}

\maketitle

In four-dimensional conformal field theory, the numbers $a$ and $c$ are defined as coefficients of particular terms in the operator product expansion (OPE) of the energy-momentum tensor.
They can be understood as conformal anomalies appearing in the trace of the energy-momentum tensor in a background gravitational field as in
\[
\<T_\mu^\mu\> = -\frac{a}{16 \pi^2} \left(\text{Euler}\right)  + \frac{c}{16 \pi^2} \left(\text{Weyl}\right)^2 
\]
where $(\text{Weyl})^2$ is the square of the Weyl curvature and $(\text{Euler})$ is the Euler density.
Generally, conformal (or Weyl) anomalies exist in any even dimension; a complete classification has appeared in \cite{Duff,Deser:1993yx}.
The quantities $a,c$ are informative invariants of a four-dimensional conformal field theory.
For example, reminiscent of Zamolodchikov's $c$-theorem in two-dimensions \cite{Zamolodchikov} is the famous ``$a$-theorem'' in four-dimensions proposed in \cite{Cardy:1988cwa} and proved in \cite{Komargodski:2011vj}.
This states that the function $a$ is monotonic along (unitary) $RG$ flow implying, in particular, that the IR value of $a$ is strictly less than its UV value.

The numbers $a,c$ are also well-defined for superconformal field theories in four dimensions, but turn out to be much simpler to characterize.
This is because the superconformal symmetry necessarily mixes spacetime symmetry with supersymmetry and hence $a,c$ admit descriptions in terms of supersymmetric invariants of the theory.
Concretely, as we will review below, the $a,c$ anomalies can be understood in terms of anomalies to $R$-symmetry.

The $R$-symmetry is also used to twist a supersymmetric theory \cite{WittenTwist,CostelloHol,ESW}.
Typically, a supersymmetric theory is defined on some geometric class of spin manifolds.
The $R$-symmetry allows one to re-label fields of a theory so as to make sense of the theory in other geometric contexts.
So, for a topological twist, one can use $R$-symmetry to define the theory on an arbitrary smooth manifold.
For a holomorphic twist \cite{CostelloHol}, one can use $R$-symmetry to define the theory on complex manifolds.
$R$-symmetry is not the only thing needed to twist a supersymmetric theory; additionally, one must choose the data of a square-zero supercharge.
Fields and operators in the twisted theory can be described in terms of the cohomology with respect to the chosen supercharge.
The twisting process, hence, often provides a significant simplification of the original supersymmetric model.
The holomorphic twist is especially relevant for this paper since we consider four-dimensional supersymmetric theories which do not, generally, admit topological twists.

In this paper, we give an interpretation of the $a,c$ anomalies in terms of complex geometric invariants of the \textit{holomorphic twist} of a four-dimensional supersymmetric quantum field theory.
The main result is that $a,c$ can be understood, in the twist, in terms of anomalies to holomorphic diffeomorphism symmetry.
Since twisting mixes $R$-symmetry with spacetime symmetry, we find that while this interpretation is not surprising, it does provide a satisfying complex geometric perspective for the coefficients which is also computationally effective via the Riemann--Roch theorem.

As a consequence of the holomorphic formulation of the $a,c$ anomalies one obtains a classification of complex two-dimensional analogs of the Virasoro algebra in conformal field theory.
The higher dimensional affine Kac--Moody Lie algebra was proposed by Faonte, Hennion, and Kapranov in \cite{FHK} and was connected to symmetries of holomorphic quantum field theory in \cite{GWkm}.
The key point is that it is absolutely necessary, in higher dimensions, to work in a derived setting.
Indeed, the $n$-dimensional Kac--Moody algebras are all obtained as central extensions of the dg Lie algebra of derived sections of the structure sheaf on the punctured $n$-disk with values in an ordinary Lie algebra $\lie{g}$.
Similarly, one can contemplate central extensions of the the dg Lie algebra of derived sections of the tangent sheaf of the punctured $n$-disk.
In complex dimension two, $a,c$ label all possible central extensions of $\Gamma(\Hat{D}^2 - 0, \T)$.
For an approach to characterizing central extensions of the derived tangent sheaf which uses factorization homology we refer to \cite{HKgf}.
While we don't pursue any algebraic structures in this paper, we refer to recent work of Bomans and Wu \cite{Bomans:2023mkd} for a characterization of the two-dimensional Virasoro algebra from the point of view of the holomorphic twist of four-dimensional supersymmetery.
We will return to the relationship between the $a,c$ anomalies and the higher dimensional Virasoro algebra from the point of view of factorization algebras in future work.

The key technical aspect of this paper is a formulation of the supersymmetric theory, its twist, and symmetries within the Batalin--Vilkovisky (BV) formalism.
The universal symmetry of a holomorphic field theory is introduced in section \ref{s:vf}.
In section \ref{s:bv}, we provide a background of anomalies in the BV formalism, specifically in the context of symmetries which largely follows Costello and Gwilliam's treatment \cite{CG2}. 
We turn to holomorphic theories in section \ref{s:hol}, and interpret the anomaly to quantizing holomorphic diffeomorphism symmetry as the Chern class of a line bundle over the universal complex manifold; thus affording the Riemann--Roch theorem as an effective method of computation.
In section \ref{s:rsymmetry} we turn to twists of supersymmetric theories and prove the precise relationship between anomalies to holomorphic diffeomorphism with $a,c$.
Finally, in section \ref{s:qcd} we elaborate on an example involving the twist of supersymmetric quantum chromodynamics (QCD).

\subsection*{Acknowledgements}

I thank Kevin Costello and Owen Gwilliam for their development of the quantum Noether theorem within the Batalin--Vilkovisky formalism as presented in their book \cite{CG2} which we rely heavily on throughout this paper.
I also am grateful to Ingmar Saberi for his collaboration on twists of four-dimensional $\cN=1$ superconformal theories which was the catalyst for this paper.
I thank Jingxiang Wu for feedback on a draft of this paper. 
Finally, I thank Boston University and the Center for Advanced Studies (CAS) at the Ludwig Maximilian University (LMU) in Munich, Germany for hospitality during the preparation of this work. 

\tableofcontents

\newpage

\section{A local model for holomorphic vector fields}
\label{s:vf}

Let $X$ be a complex manifold.
The $\dbar$-operator associated to a holomorphic vector bundle $V$ defines its Dolbeault complex $\Omega^{0,\bu}(X , V)$.
In degree $q$, the space $\Omega^{0,q}(X, V)$ denotes the space of smooth sections of the vector bundle $\Wedge^q \br \T_X^* \otimes V$ and the $\dbar$-operator extends to a differential $\dbar \colon \Omega^{0,q}(X,V) \to \Omega^{0,q+1}(X,V)$, $\dbar^2 = 0$.
%For any $U \subset X$, the complex $\Omega^{0,\bu}(U,V)$ is defined; in this way the Dolbeault complex is naturally a sheaf of cochain complexes.
%In fact, $\Omega^{0,\bu}(X,V)$ is an elliptic complex, see appendix \ref{a:local}.
By Dolbeault's theorem, the complex provides a resolution for the sheaf of holomorphic sections of the bundle $V$.

A holomorphic vector field is a holomorphic section of the holomorphic tangent bundle $\T_X$.
The sheaf of holomorphic vector fields has the natural structure of a sheaf of Lie algebras with bracket the Lie bracket of vector fields.
Since this bracket involves only holomorphic differential operators, it extends to a bracket on the Dolbeault complex of $\T_X$.
In local holomorphic coordinates the Lie bracket is
\beqn
[f_{I,i} (z,\zbar) \d \zbar_I \del_{z_i} , g_{J,j} (z,\zbar) \d \zbar_J \del_{z_j}] = f_{I,i} \del_{z_i} g_{J,j} \d \zbar_I \wedge \d \zbar_J \del_{z_j} - g_{J,j} \del_{z_j} f_{I,i} \d \zbar_J \wedge \d \zbar_I \del_{z_i}.
\eeqn
This bracket endows $\Omega^{0,\bu}(X, \T_X)$ with the structure of a dg Lie algebra where the differential is $\dbar$.
This leads to the following definition; for a recollection of terminology on local complexes and local Lie algebras we refer to \cite{CG1}.

\begin{dfn}
Let $\cT_X$ be the local dg Lie algebra whose underlying complex of vector bundles is
\beqn
\Omega^{0,\bu}(X, \T_X) .
\eeqn
The differential is the $\dbar$-operator associated to the holomorphic bundle $\T_X$.
The bracket is the Lie bracket of holomorphic vector fields extended to the Dolbeault complex as described above.

When $X = \C^n$ we will abusively refer to this local dg Lie algebra simply as $\cT$.
\end{dfn}

By Dolbeault's theorem, for $U \subset X$ a Stein open set there is a quasi-isomorphism $\cT_X(U) \simeq \op{Vect}^{hol}(U)$, the Lie algebra of holomorphic vector fields on $U$.

There is also the following relationship between $\cT$ and formal vector fields, which will be used subsequently.
Let \(E \to M\) denote a $\Z$-graded vector bundle on \(M\).
We consider the pro vector bundle of $\infty$-jets which we will denote by $j^\infty E$, see \cite{Anderson} or \cite[\S 5.6]{CostelloBook} for instance.
The sheaf of smooth sections of this pro vector bundle carries the natural structure of a $D_M$-module.
If $\cL = \Gamma(L)$ is a local $L_\infty$ algebra then $j^\infty L$ is a bundle of $L_\infty$ algebras.

\begin{lem}
The map
\beqn
j_0^\infty \colon \cT(\C^n) \to \lie{vect}(n) [[\zbar, \d \zbar]] 
\eeqn
induces a quasi-isomorphism of dg Lie algebras $j_0^\infty \cT (\C^n) \simeq \lie{vect}(n)$.
\end{lem}

\section{Anomalies and descent}
\label{s:bv}

In this section we rely on the quantum Batalin--Vilkovisky formalism as developed in the series of books \cite{CG1,CG2}.

\subsection{Anomalies as local functionals}
\label{sec:anomalies1}

Let $E$ be a vector bundle on a smooth manifold $M$.
The $C^\infty_M$-module of \defterm{local functionals} on $M$ is
\beqn
\oloc(E) \define {\rm Dens}_X \otimes_{D_M} \prod_{n > 0} {\rm Hom}_{C^\infty_M} \left(\op{jet} (E) , C^\infty_M\right) ,
\eeqn
where $D_M$ is the algebra of differential operators and $\op{jet}(E)$ is the $\infty$-jet bundle of~$E$.
Concretely, if $M$ is equipped with a volume form, a section of $\oloc(E)$ is a linear combination of objects of the form
\[
(\phi_1,\ldots,\phi_n) \mapsto \left[ (D_{1} \phi_1 \cdots  D_{n} \phi_n)  \; \dvol_M\right]
\]
where $\phi_i$ denote sections of $E$, where $D_i$ are differential operators $\sE \to C^\infty(M)$, and where $[-]$ means equivalence classes up to total derivative.

Suppose that $(\cE,\omega, S)$ is a classical BV theory on a smooth manifold $M$.
This means, in particular:
\begin{itemize}
\item $\cE = \Gamma(M,E)$ is sections of a $\Z$-graded vector bundle $E \to M$.
\item $\omega \colon E \otimes E \to \op{Dens}_M [-1]$ is a density valued, non-degenerate, pairing of degree~$-1$.
\item $S \in \oloc(\cE)$ is a local functional of cohomological degree zero.
\end{itemize}

The BV anti-bracket induced from $\omega$ equips the complex $\oloc(\cE)[-1]$ with the structure of a graded Lie algebra.
The \textit{classical master equation} is the condition that $S$ be a Maurer--Cartan element $\{S,S\} = 0$.
Equivalently, we can decompose $\{S,-\} = Q + \{I,-\}$ where $S(\Phi) = \omega(\phi,Q\phi) + I(\phi)$.
Then the classical master equation is $Q I + \frac12 \{I,I\} = 0$.
When this equation holds, the operator $\{S,-\}$ is square-zero and endows
\beqn
\left(\oloc(\cE), \{S,-\}\right) ,
\eeqn
with the structure of a dg Lie algebra. 

Our goal in the remainder of the section is recall the various versions of symmetry that a theory in the BV formalism can exhibit.
To that end, suppose that $\lie{g}$ is an arbitrary dg Lie algebra.
Then we can define its Chevalley--Eilenberg complex $C^\bu(\lie{g})$ which is a commutative dg algebra.
Furthermore, we can tensor with a shift of local functionals
\beqn
C^\bu(\lie{g}) \otimes \oloc(\cE)[-1]
\eeqn
to obtain a dg Lie algebra.
The subcomplex 
\beqn\label{eqn:red}
C^\bu_{red}(\lie{g}) \otimes \oloc(\cE)[-1]
\eeqn
is a sub dg Lie algebra.
To write this in a slightly more symmetric way, we recall that $S$ equips $\cE[-1]$ with the structure of a local $L_\infty$ algebra such that $\{S,-\}$ is precisely the Chevalley--Eilenberg differential.
Thus we can write this sub dg Lie algebra as
\beqn\label{eqn:red}
C^\bu_{red}(\lie{g}) \otimes \cloc^\bu(\cE[-1])[-1] .
\eeqn

Let $\cL$ be a local dg Lie algebra (our main example is the local dg Lie algebra $\cT_X$ from the previous section).
The Chevalley--Eilenberg differential equips the $C^\infty_M$-module of local functionals $\oloc(\cL[1])$ on the shift of $\cL$ with the structure of a cochain complex which we denote $\cloc^\bu(\cL)$
There is a cochain embedding
\beqn
\cloc^\bu(\cL)(M) \subset C_{red}^\bu(\cL(M)) 
\eeqn
from local cochains (by our conventions, local cochains are always reduced) to ordinary Chevalley--Eilenberg cochains of the dg Lie algebra of global sections $\cL(M)$.

Furthermore, the Lie bracket in \eqref{eqn:red} restricts to a Lie bracket on the subcomplex
\beqn
\cloc^\bu(\cL \oplus \cE[-1])[-1] \subset C^\bu_{red}(\lie{g}) \otimes \cloc^\bu(\cE[-1])[-1] ,
\eeqn
thus equipping the left hand side with the structure of a dg Lie algebra.

The embedding of local dg Lie algebras $\cE[-1] \to \cL \oplus \cE[-1]$ induces a map of dg Lie algebras
\beqn
\cloc^\bu(\cL \oplus \cE[-1])[-1] \to \cloc^\bu(\cE[-1]) [-1] = \oloc(\cE)[-1] .
\eeqn
Let $\op{InnerAct}(\cL,\cE)$ denote the kernel of this map.
The map of local dg Lie algebras $\cL \oplus \cE[-1] \to \cL$ induces an inclusion of dg Lie algebras
\beqn
\cloc^\bu(\cL) \to \cloc^\bu(\cL \oplus \cE[-1])[-1] .
\eeqn
This map factors through the dg Lie algebra $\op{InnerAct}(\cL,\cE) \subset \cloc^\bu(\cL \oplus \cE[-1])[-1]$.
We let 
\beqn
\op{Act}(\cL,\cE) \define \op{InnerAct}(\cL,\cE) \slash \cloc^\bu(\cL) 
\eeqn
be the quotient dg Lie algebra.

As cochain complexes there are isomorphisms
\begin{align*}
\op{Act}(\cL,\cE) & \simeq \cloc^\bu(\cL \oplus \cE[-1])[-1] \slash \left(\cloc^\bu(\cL)[-1] \oplus \cloc^\bu(\cE[-1])[-1] \right) \\
\op{InnerAct}(\cL,\cE) & \simeq \cloc^\bu(\cL \oplus \cE[-1])[-1] \slash \cloc^\bu(\cE[-1])[-1] .
\end{align*}
Thus, an element of $\op{Act}(\cL,\cE)$ is a local functional on the complex $\cL[1] \oplus \cE$ which does not depend solely on $\cL[1]$ or $\cE$.
Likewise, an element of $\op{InnerAct}(\cL,\cE)$ is a local functional which does not depend solely on $\cE$.

\begin{dfn}
Let $(\cE,\omega, S)$ be a classical BV theory and $\cL$ a local Lie algebra.
A \defterm{classical $\cL$-background} is a Maurer--Cartan element
\beqn
S^{\cL} \in \text{Act}(\cL,\cE)^1 .
\eeqn
Unwrapping the definitions above, this means that $\cS^{\cL}$ is a local functional on $\cL[1] \oplus \cE$ of cohomological degree zero which satisfies the equation
\beqn\label{eqn:equivariantmaster}
\d_{\cL} S^{\cL} + \{S, S^{\cL}\} + \frac12 \{S^{\cL},S^{\cL}\} = 0 ,
\eeqn
where $\d_{\cL}$ is the Chevalley--Eilenberg differential for $\cL$.
\end{dfn}

The equation \eqref{eqn:equivariantmaster} is referred to in \cite{CG2} as the $\cL$-\textit{equivariant classical master equation.}

In \cite{CostelloBook} the problem of perturbative quantization of classical theories within the BV formalism is put within an obstruction-deformation theoretic framework.
There are two main steps to construct a quantization.
The first is to construct a renormalized, or effective action, $\{I[L]\}$ which is a family of action functionals depending on a scale parameter $L$ whose $L \to 0$ limit agrees with the classical interaction $I$ (among other conditions).
Additionally, this effective family is required to satisfy the \textit{quantum master equation}
\beqn\label{eqn:qme}
(Q + \hbar \triangle_L) e^{I[L]/\hbar}  = 0 ,
\eeqn
where $\triangle_L$ is the effective BV Laplacian.

Starting with a classical BV theory $(\cE,\omega,S = S_{kin} + I)$, the problem of quantization order by order in $\hbar$ is controlled by the dg Lie algebra $\oloc(\cE)[-1]$.
While there are possibly many inequivalent ways to quantize a theory, even at one-loop, there is a scheme-dependent obstruction to the one-loop quantization of a classical BV theory which is measured by a degree $+1$ cohomology class
\beqn
\Theta_{1-loop} \in H^1_{loc}(\cE) = H^2 \left(\oloc(\cE)[-1], \{S,-\}\right) ,
\eeqn
obtained as the $L \to 0$ limit of the quantum master equation (in other words, it is the failure for an effective action to solve the master equation).
Similarly, if a quantization exists and is fixed to order $\hbar^{n-1}$ there is an obstruction $\Theta_{n-loop}\in H^1_{loc}(\cE)$ to extending this to a quantization to order $\hbar^n$.

Now we turn to the question of the equivariant quantization; more specifically, the quantization of $\cL$-backgrounds as defined above.
Here is the summary of the definition, we refer to \cite{CG2}[Definition 13.2.2.1] for the full definition.

\begin{dfn}
Suppose $\{I[L]\}$ is an effective family of interactions describing a quantum field theory (satisfies, amongst other conditions, the ordinary quantum master equation \eqref{eqn:qme}).
A \defterm{quantum $\cL$-background} is a family of effective actions
\beqn
I^{\cL}[L] \in \cO(\cL[1] \oplus \cE) \slash \cO(\cL[1]) [[\hbar]]
\eeqn
which agrees with $I[L]$ modulo functionals of $\cL[1]$ and
satisfies the following equivariant master equation
\beqn
\d_\cL I^{\cL}[L] + Q I^{\cL}[L] + \frac12 \{I^{\cL}[L],I^{\cL}[L]\}_L + \hbar \triangle_L I^{\cL}[L] = 0 
\eeqn
in $\cO(\cL[1] \oplus \cE) \slash \cO(\cL[1]) [[\hbar]]$.

Likewise, an \defterm{inner} quantum $\cL$-background is an effective family in a bigger class of functionals 
\beqn
I^{\cL}[L] \in \cO(\cL[1] \oplus \cE) [[\hbar]]
\eeqn
satisfying the same quantum master equation.
\end{dfn}

In the equivariant context the dg Lie algebras $\op{Act}(\cL,\cE)$ and $\op{InnerAct}(\cL,\cE)$ control different equivariant quantization situations, analogously to how the dg Lie algebra $\oloc(\cE)[-1]$ controls the non-equivariant situation.
For simplicity, we consider the problem of quantization to one-loop.
Throughout, we assume that a fixed quantization of the classical BV theory $(\cE,\omega,Q)$ is fixed at one-loop (in particular the one-loop obstruction $\Theta_{1-loop}$ vanishes).
The complex which controls quantizations of a classical $\cL$-background while leaving this quantization of the BV theory fixed is precisely~$\op{Act}(\cL,\cE)$.

\begin{prop}
Let $S^{\cL} \in \op{Act}(\cL,\cE)$ be a classical $\cL$-background and suppose $S_{1-loop} = S + \hbar S^{(1)}$ is a one-loop quantization of a classical BV theory $(\cE,\omega, S)$.
The obstruction to lifting this to a one-loop $\cL$-background is an element
\beqn
\Theta^{\cL,\cE} \in H^2 \left(\op{Act}(\cL,\cE) \right) ,
\eeqn
which, in this situation, lifts to an element
\beqn
\Theta^{\cL,\cE} \in H^1_{loc}\left(\cL \ltimes \cE[-1] \; | \; \cL\right) .
\eeqn

Similarly, the obstruction to lifting this to a one-loop inner $\cL$-background is an element
\beqn
\Theta^{\cL} \in H^2 \left(\op{InnerAct}(\cL,\cE) \right) 
\eeqn
which, in this situation, lifts to an element 
\beqn
\Theta^{\cL} \in H_{loc}^1 (\cL) = H^2 \left(\cloc^\bu(\cL)[-1] \right) .
\eeqn
\end{prop}

\subsection{Central charges}

The local cohomology classes of cohomological degree one deserve a special name which takes the form of the following terminological definition.

\begin{dfn}
The \defterm{space of central charges} for a local Lie algebra $\cL$ is the vector space 
\beqn
H_{loc}^1 (\cL) .
\eeqn
In particular, an obstruction to a quantum inner $\cL$-background determines a central charge for $\cL$.
\end{dfn}

The motivation for this definition is the following.
Suppose that our local Lie algebra $\cL = \Gamma(L)$ is defined on $\R^d$.
Typically, a \textit{current} is a $(d-1)$-form valued in Lagrangian densities.
More generally, for $\cL$ a local Lie algebra, we obtain the space of currents valued in $\cL_c$ by evaluating on the codimension one sphere $S^{d-1}$.
Here $\cL_c = \Gamma_c(L)$ is the cosheaf of compactly supported sections of the underlying vector bundle $L$.
Some care must be taken in this definition, since strictly it only makes sense to evaluate $\cL$ on open subsets; thus one should consider the value of $\cL$ on open sets which are infinitesimal neighborhoods of the $(d-1)$-sphere.
In nice situations, which we will turn to momentarily, the value of $\cL$ on a neighborhood of the form $(-\epsilon,\epsilon) \times S^{d-1}$ can be written as
\beqn\label{eqn:splitting}
\Omega^\bu_c (- \epsilon, \epsilon) \otimes \lie{g}_{S^{d-1}}
\eeqn
where $\lie{g}_{S^{d-1}}$ is some dg Lie (or more generally $L_\infty$) algebra.
We call $\lie{g}_{S^{d-1}}$ the Lie algebra of $S^{d-1}$-currents associated to $\cL$.

Now suppose that $\phi \in C^{1}_{loc}(\cL)$ is a local cocycle of cohomological degree $+1$.
Assume its value on $(-\epsilon,\epsilon) \times S^{d-1}$ can be expressed as $\int_{-\epsilon}^{\epsilon} \boxtimes \psi$ where
\begin{itemize}
\item $\int_{-\epsilon}^{\epsilon} \colon \Omega^\bu_c (- \epsilon, \epsilon) \to \C[-1]$ is the (degree one) integration map along the interval, and
\item $\psi$ is a degree two cocycle for the Lie algebra $\lie{g}_{S^{d-1}}$.
Automatically, $\psi$ is of cohomological degree two since integration along the interval is of cohomological degree one.
\end{itemize}
In this situation, $\psi$ determines a central extension 
\beqn
0 \to \C \to \Hat{\lie{g}}_{S^{d-1}} \to \lie{g}_{S^{d-1}} \to 0 ,
\eeqn
and we interpret $\lie{g}_{S^{d-1}}$ as the Lie algebra of centrally extended currents associated to $\cL$ and $\phi$.

We consider an example, studied in \cite{GWkm}.
Let $\lie{g}$ be an ordinary Lie algebra and consider the following local Lie algebra defined on $\C^n = \R^{2n}$
\beqn
\cL = \lie{g} \otimes \Omega^{0,\bu}(\C^n) .
\eeqn
Take the neighborhood $\C^n - 0 = \R \times S^{2n-1}$ of the sphere.
In this situation, there is a dense embedding of $\lie{g}$-valued derived algebraic sections of the structure sheaf into the spherical currents
\beqn
\lie{g} \otimes \R \Gamma(D^n - 0, \cO) \hookrightarrow \lie{g}_{S^{2n-1}} .
\eeqn
Any invariant and symmetric homogenous polynomial $\theta \in \Sym^{n+1}(\lie{g})^{G}$ determines a local cocycle for $\lie{g} \otimes \Omega^{0,\bu}(\C^n)$ defined by
\beqn
\phi(\alpha_0,\ldots,\alpha_n) = \int_{\C^n} \theta ( \alpha_0 \del \alpha_1\cdots \del \alpha_n ) .
\eeqn
The corresponding cocycle $\psi$ is precisely the higher dimensional Kac--Moody cocycle \cite{FHK}:
\beqn
\psi(f_0,\ldots,f_n) = \oint_{S^{2n-1}} \theta ( \alpha_0 \del \alpha_1\cdots \del \alpha_n ) 
\eeqn
where $f_0,\ldots,f_n \in \lie{g} \otimes \R \Gamma(D^n - 0, \cO)\subset \lie{g}_{S^{2n-1}}$.
Thus, every choice of ``level'' $\theta$ defines a central extension of the spherical currents.
In \cite{GWkm} we prove that such $\theta$ label all such central charges for this local Lie algebra which are $GL(n)$ and translation invariant.

Similarly, one could start with the local Lie algebra $\cT$ resolving holomorphic vector fields on $\C^n$.
We will explain a characterization of central charges in this case in the next section.
We will not, however, return to the resulting higher dimensional Virasoro algebras such central charges determine.
We refer to \cite{Bomans:2023mkd} for results in the complex two-dimensional setting.

\subsection{Diffeomorphism anomalies and Gelfand--Fuks cohomology}

We turn our attention to the holomorphic version of diffeomorphism, or gravitational, anomalies.
As recollected, the underlying local Lie algebra describing infinitesimal holomorphic diffeomorphisms is $\cT$, which resolves the Lie algebra of holomorphic vector fields.

The following result characterizes, in any dimension, the local anomalies to holomorphic reparametrization invariance.

\begin{thm}
\cite{BWgf}
\label{thm:bwgf}
On flat space $X = \C^n$ there is a quasi-isomorphism
\beqn
\cloc^\bu \left(\cT ( \C^n ) \right) \simeq C^\bu_{red}(\lie{vect}(n)) [2n] .
\eeqn
The space of central charges (on flat space) for $\cT$ is
\beqn
H_{loc}^1 (\cL (\C^n)) \simeq H^{2n+1}(\lie{vect}(n)) .
\eeqn
\end{thm}

Consider the one-dimensional case $n=1$.
Then the space of central charges for holomorphic vector fields is $H^3(\lie{vect}_1)$, which is one-dimensional.
By transgression this space is naturally isomorphic to the cohomology of the holomorphic Witt algebra $\lie{witt} = \C((z)) \del_z$, which is the Lie algebra of vector fields on the punctured formal disk.
The unique (up to scale) central charge for the local Lie algebra $\cT(\C)$ corresponds to the usual central charge in conformal field theory.
As a local cocycle for $\cT (\C)$, a local representative for this central charge is
\beqn
\phi(\mu, \mu') = \int_\C J \mu \wedge \del (J \mu') .
\eeqn
Of course, in coordinates $\mu = f(z,\zbar) \del_z$, where $f(z,\zbar)$ is a Dolbeault form of type $(0,\bu)$, this recovers the familiar expression for the Virasoro cocycle.

In the notation of the previous section, the space of spherical currents in this example receives a dense embedding of the Lie algebra of algebraic vector fields on the punctured disk $\C[z,z^{-1}] \del_z \hookrightarrow \lie{g}_{S^1}$.
The corresponding cocycle for $\C[z,z^{-1}]\del_z$, denoted $\psi$ above, is the standard Virasoro cocycle.

We comment that the result above can be generalized.
In \cite{BWgf}, see also \cite{BWthesis}, it is shown that for \textit{any} complex manifold $X$ of complex dimension $n$ that there is an isomorphism
\beqn
H_{loc}^\bu(\cT(X)) \simeq H^\bu_{dR}(X) \otimes H^\bu (\lie{vect}(n)) [2n] .
\eeqn
When $X = \C^n$ we recover the result as stated above.

\subsection{The holomorphic $a,c$ anomalies}

In the last section we have seen how the space of central charges for the local Lie algebra of holomorphic vector fields can be expressed in terms of the Gelfand--Fuks cohomology of formal vector fields.
In this section we recall an elegant description of this cohomology, following \cite{Fuks} and match with explicit representatives of the corresponding local cocycles in complex dimension two.

Let $\op{Gr}(n,k)$ be the complex Grassmannian of $n$-planes in $\C^k$. 
Denote by $\op{Gr}(n,\infty)$ the colimit of the natural sequence of topological spaces
\beqn
\{\star\} \to {\rm Gr}(n, n+1) \to \op{Gr}(n,n+2) \to \cdots . 
\eeqn 
It is a standard fact that ${\rm Gr}(n, \infty)$ is a model for the classifying space $BGL(n)$ of rank $n$ vector bundles.
%Let $EU(n) \to BU(n)$ be the universal principal $U(n)$-bundle. 
From the colimit description, this results in the skeletal filtration of the classifying space:
\beqn
{\rm sk}_{k} BGL(n) = {\rm Gr}(n,k) .
\eeqn 
Let $X_n$ denote the restriction of the universal bundle $EGL(n)$ over the $2n$-skeleton:
\beqn
\begin{tikzcd}
X_n \ar[r] \ar[d] & E GL(n) \ar[d] \\
{\rm sk}_{2n} B GL(n) \ar[r] & B GL(n) .
\end{tikzcd}
\eeqn

%\begin{rmk}
%Though not the way the Gelfand and Fuks originally proved the result, one can use the computation of the cohomology of $\W_d$ with coefficients in $\hOmega^k_d$ together with the spectral sequence (\ref{ss1}) to prove this description of $H^*(\W_d)$. 
%Indeed, the spectral sequence (\ref{ss1}) is isomorphic, up to regradings, to the Serre spectral sequence for the principal $\U(d)$-bundle $X_d \to {\rm sk}_{2d} B \U(d)$. 
%In other words, the formal de Rham differential on $\hOmega^*_d$ is exactly the $E_2$ differential for the Serre spectral sequence. 
%\end{rmk}

\begin{thm}[\cite{Fuks} Theorem 2.2.4] 
\label{thm:fuks}
There is an isomorphism of graded commutative algebras
\beqn
H^\bu(\lie{vect}(n)) \cong H^\bu_{dR} (X_n) .
\eeqn
The commutative product is the trivial one.
\end{thm}

Note that when $n = 1$ we have ${\rm sk}_2 B U(1) = \PP^1 \subset \PP^\infty = B U(1)$. 
Moreover, the restriction of the universal bundle is the Hopf fibration $U(1) \to S^3 \to \PP^1$ so that $X_1 \simeq S^3$.

It will be useful to say a few words about the proof of this theorem, as we will need a more general version to discuss mixed anomalies shortly.

For a Lie algebra $\lie{g}$ with a subalgebra $\lie{h} \subset \lie{g}$ the relative Lie algebra cohomology is defined to be the cohomology of the subcomplex 
\beqn
C^\bu(\lie{g} \; | \; \lie{h}) \subset C^\bu(\lie{g})
\eeqn
which consists of $\lie{h}$-invariant cochains $\phi \in C^\bu(\lie{g})$ satisfying $\iota_X \phi = 0$.
Here $\iota_X \colon C^\bu(\lie{g}) \to C^{\bu-1} (\lie{g})$ is the contraction with an element $X \in \lie{h}$.
The Hochschild--Serre spectral sequence for a subalgebra $\lie{h} \subset \lie{g}$ is a spectral sequence whose $E_2$ page is given in terms of relative Lie algebra cohomology
\beqn
E_2 \simeq H^\bu(\lie{g}, \lie{h}) \otimes H^\bu(\lie{h}),
\eeqn
and converges to the absolute Lie algebra cohomology $H^\bu(\lie{g})$.

In the current context we are interested in the absolute cohomology of $\lie{vect}(n)$.
The key idea is that the Hochschild--Serre spectral sequence for the subalgebra $\lie{gl}_n \subset \lie{vect}(n)$ consisting of linear vector fields can be interpreted geometrically in terms of classifying spaces.
One can identify the relative Lie algebra cohomology $H^\bu(\lie{vect}(n) \; | \; \lie{gl}_n)$ with the cohomology of the $2n$-skeleton $\op{sk}_{2n} BGL(n)$.
This identification can be extended to an isomorphism of spectral sequences between the Serre spectral sequence computing the cohomology of $X_n$, as a $GL(n)$ space over the $2n$-skeleton, and the Hochschild--Serre spectral sequence for the subalgebra $\lie{gl}_n \subset \lie{vect}(n)$ converging to the absolute cohomology of $\lie{vect}(n)$.

An immediate application of the Serre spectral sequence shows that first nontrivial cohomology above degree zero of $X_n$ is in degree $2n+1$ and we have
\beqn
H^{2n+1}(\lie{vect}(n)) \simeq H^{2n+1}(X_n) \simeq H^{2n+2}(BGL(n)) .
\eeqn
In other words, the degree $2n+1$ cohomology of formal vector fields is isomorphic to degree $2n+2$ polynomials built from universal characteristic classes of rank $n$ vector bundles.

In the previous section we have argued that the space of holomorphic diffeomorphism anomalies in $n$-dimensions, the space of central charges for $\cT$, is precisely this first non-trivial cohomology group of $X_n$.
When $n=1$ we have that $X_1 \simeq S^3$ hence the space of central charges for holomorphic vector fields in complex dimension one is one-dimensional, as expected (all being multiples of the universal characteristic class is $c_1^2 \in H^4 (BGL_1)$).
When $n = 2$ we see that the space of central charges is two-dimensional
\beqn
H^1_{loc}(\cT(\C^2)) \simeq H^{5} (\lie{vect}_2) \simeq H^6(BGL(2)) = \C \cdot c_1^3 \oplus \C \cdot c_1 c_2
\eeqn
spanned by the universal characteristic classes $c_1^3$ and $c_1 c_2$.

It will be convenient for our purposes to relabel this cohomology in terms of the degree six components of the universal Chern character:
\beqn
\ch_1^3 = c_1^3, \quad \ch_1 \ch_2 = \frac12 \left(c_1^3 - 2 c_1 c_2\right) .
\eeqn 
Any holomorphic diffeomorphism anomaly represented by a class $\Theta \in H^5(\lie{vect}(n))$ can be written as a linear combination of these universal classes
\beqn
\label{eqn:achol}
\Theta = a^{hol} \cdot \ch_1 \ch_2 + c^{hol} \cdot \ch_1^3 .
\eeqn
We will see momentarily how the coefficients $a^{hol},c^{hol}$ are related to the standard $a,c$ anomalies in four-dimensional conformal geometry.

We state explicit local models for the central charges $\ch_1 \ch_2, \ch_1^3$.

\begin{prop}\label{prop:holac}
The following are explicit degree one local cocycle representatives for the classes $\ch_1 \ch_2$ and $\ch_1^3$:
\beqn
\Theta_{a^{hol}} = \frac{1}{12} \int_{\C^2} \op{tr} ( J \mu ) \op{tr}(\del J \mu \del J \mu) ,
\eeqn
and
\beqn
\Theta_{c^{hol}} = \frac{1}{6} \int_{\C^2} \op{tr}(J \mu) \op{tr}(\del J \mu) \op{tr}(\del J \mu) 
\eeqn
respectively.
\end{prop}
\begin{proof}
The arguments that follow are elaborated in \cite{BWgf}. 

There are two main steps.
The first is to identify explicit representatives in the cohomology of formal vector fields.
The second is in passing from Gelfand--Fuks classes in $H^5(\lie{vect}(2))$ to local cohomology classes in $H^1_{loc}(\cT(\C^2))$.
Both steps follow the general principle of descent.

The first step is accomplished by considering the double complex associated to the Lie algebra cohomology of $\lie{vect}(2)$ with coefficients in the formal de Rham complex~$\Omega^\bu(\Hat{D}^2)$.
Since the de Rham complex is quasi-isomorphic to $\C$ in degree zero, the resulting spectral sequence converges to $H^\bu(\lie{vect}(2))$.
The first page of the spectral sequence is generated, as an bigraded algebra, by generators $a_1,a_2$ of bidegree $(0,1),(0,3)$ respectively and by generators $(\tau_1,\tau_2)$ of bidegree $(1,1)$, $(2,2)$ respectively.
These generators are subject to the relations $\tau_1^3 = \tau_1 \tau_2 = \tau_2^2 = 0$.
In degree five, the following two classes remain at the $E_\infty$ page: $a_1 \tau_1^2, a_1 \tau_2$ which are represented by the totally skew symmetric tri-linear functionals on $\lie{vect}(2)$ with values in $\Omega^2(\Hat{D}^2)$:
\begin{align*}
a_1 \tau_1^2 \colon (X_0,X_1,X_2) & \mapsto \op{tr}(J X_0) \op{tr} (\d J X_1 \wedge \op{tr} (\d J X_2) + \cdots\\
a_1 \tau_2 \colon (X_0,X_1,X_2) & \mapsto \op{tr}(J X_0) \op{tr} (\d J X_1 \wedge \d J X_2) + \cdots .
\end{align*}
The $\cdots$ denote appropriate skew-symmetrization.

The passage from Gelfand--Fuks classes to local cohomology classes uses descent for the $\dbar$-operator.
If $\phi$ is one of the representatives above, we can extend it to a functional on $\cT = \Omega^{0,\bu}(\C^2, \T)$ with values in densities on $\C^2$ and can then take the local class.
The resulting local functional is of total degree $+1$ since the functionals are cubic and must have anti-holomorphic Dolbeault degree which sums to two.
\end{proof}

\subsection{Mixed anomalies}

In this section we generalize the discussion of holomorphic diffeomorphism anomalies on $\C^n$ to anomalies for a local Lie algebra built from both infinitesimal holomorphic diffeomorphisms and infinitesimal holomorphic gauge symmetries.
For gauge symmetries of the trivial holomorphic $G$-bundle, the local Lie algebra is the semi-direct product local Lie algebra
\beqn
\cT_X \ltimes \lie{g} \otimes \Omega^{0,\bu}_X .
\eeqn
Where the local Lie algebra $\cT_X$ resolving holomorphic vector fields acts by Lie derivative on the Dolbeault complex for the trivial bundle.

Our goal is to describe, just as we did in the case of $\cT_X$, the space of anomalies for the existence of quantum backgrounds with respect to this local Lie algebra.
Following the discussion of section \ref{sec:anomalies1} the space of such anomalies is the degree one local cohomology
\beqn
H^1_{loc} (\cT_X \ltimes \lie{g} \otimes \Omega^{0,\bu}_X ).
\eeqn
We will use the characterization, in particular, to describe diffeomorphism anomalies for holomorphic gauge theory in \ref{sec:mixed2}.
In analogy with the case that the group $G$ is trivial, we have the following result, which can be proved by similar methods as in theorem \ref{thm:bwgf}.

\begin{thm}
On flat space $X = \C^n$ there is a quasi-isomorphism
\beqn
\cloc^\bu \left(\cT ( \C^n ) \ltimes \lie{g} \otimes \Omega^{0,\bu}(\C^n) \right) \simeq C^\bu_{red}(\lie{vect}(n) \ltimes \lie{g} \otimes \cO_n) [2n] .
\eeqn
In particular, the space of mixed anomalies can be identified with the following Lie algebra cohomology
\beqn
H_{loc}^1 \simeq H^{2n+1}(\lie{vect}(n) \ltimes \lie{g} \otimes \cO_n) .
\eeqn
\end{thm}

Like in the case of pure holomorphic diffeomorphism anomalies, we will identify the space of mixed anomalies with certain characteristic classes.
This will be facilitated by the following result.

\begin{thm}[\cite{Khoroshkin}]
Let $G$ be a compact Lie group and $\lie{g}$ its complexified Lie algebra.
There is a graded ring isomorphism
\beqn
H^\bu(\lie{vect}(n) \ltimes \lie{g} \otimes \cO_n \; | \; \lie{gl}_n \oplus \lie{g}) \simeq H^{\bu \leq 2n} \left(B GL(n) \times BG\right) ,
\eeqn
where the right hand side is the cohomology of the $2n$-skeleton of the product of classifying spaces.
\end{thm}

In this theorem appears the relative Lie algebra cohomology.
In our context we are considering the relative cohomology with respect to the subalgebra
\beqn\label{eqn:subalg}
\lie{gl}_n \oplus \lie{g} \subset \lie{vect}(n) \ltimes \lie{g} \otimes \cO_n 
\eeqn
consisting of linear vector fields and constant $\lie{g}$-valued functions.
Using Khoroshkin's result we can proceed similarly as in the case $\lie{g}=0$ (see theorem \ref{thm:fuks} and the discussion following) to identify the absolute Lie algebra cohomology in terms of classifying spaces.
Indeed, let $Y_n$ be the restriction of the universal bundle over the $2n$-skeleton of $B GL(n) \times BG$.
As in the case $\lie{g}=0$, there is an isomorphism between the Hochschild--Serre spectral sequence for the subalgebra \eqref{eqn:subalg} and the Serre spectral sequence for the $GL(n) \times G$-space $Y_n$ over the $2n$-skeleton of $BGL(n) \times BG$.
This leads to the following.

\begin{thm}
Let $G$ be a compact connected Lie group and $\lie{g}$ is complexified Lie algebra.
There is an isomorphism of graded rings
\beqn
H^\bu(\lie{vect}(n) \ltimes \lie{g} \otimes \cO_n) \simeq H^\bu(Y_n) .
\eeqn
In particular, there is an isomorphism
\beqn
H^{2n+1}(\lie{vect}(n) \ltimes \lie{g} \otimes \cO_n) \simeq H^{2n+2} (BGL(n) \times BG) .
\eeqn
\end{thm}

Notice that when $G$ is the trivial group we have $Y_n = X_n$ and this result recovers the  result of Fuks that we recollected in theorem \ref{thm:fuks}.
From this theorem we see that the space of mixed anomalies is isomorphic to $H^{2n+2}(BGL(n) \times BG)$.
As in the case of pure holomorphic diffeomorphism anomalies, we can match such cohomology classes with explicit classes in the local cohomology of $\cT(X) \ltimes \lie{g} \otimes \Omega^{0,\bu}(X)$.

Let $\ch_j \in H^{2j}(BGL(n)), j=1,\ldots,n$ be the universal Chern characters as before.
Also, let
\beqn
\theta_k \in \Sym^{k} (\lie{g}^*)^G \simeq H^{2k}(BG)
\eeqn
denote an invariant polynomial of polynomial degree $k$.
Then the degree one local functional corresponding to the class
\beqn
\ch_j \otimes \theta_{n-j+1} \in H^{2n+2}(BGL(n) \times BG), \quad j=0,\ldots,n,
\eeqn
for example, is proportional to
\beqn
\int_{\C^n} \op{tr}(J \mu (\del J \mu)^{j-1}) \, \theta_{n-j+1}(\del A, \ldots, \del A) .
\eeqn
This can further be extended to the case $j=0$.
In this case, the local functional corresponding an invariant polynomial $\theta_{n+1}$ on $\lie{g}$ of polynomial degree $n+1$ is
\beqn
\int_{\C^n} \theta_{n+1}(A, \del A, \ldots, \del A) . 
\eeqn
This type of local cocycle, which does not involve any holomorphic vector fields, was studied in \cite{GWkm} where it is shown how it relates to higher dimensional Kac--Moody algebras \cite{FHK}.

\section{Holomorphic theories and their quantum anomalies}
\label{s:hol}

\subsection{A local index theorem from free field theory}

To any elliptic complex $(\cE,Q)$ on a manifold $M$ there exists a free theory in the Batalin--Vilkovisky formalism whose complex of fields (including ghosts, fields, anti-fields, etc.) is `double' the size of $\cE$:
\beqn
\T^* [-1] \cE = \cE \oplus \cE^![-1] ,
\eeqn
see \cite[Section ??]{CG2}.
Here $\cE^!$ stands is the elliptic complex which is given by sections of the bundle $E^* \otimes \op{Dens}_M$ and whose differential is $Q^* \otimes \id$; here $E$ is the bundle whose sections is the underlying graded vector space of the elliptic complex $\cE$.
A field is a section of $\T^*[-1] \cE$ which we denote by the tuple $(\gamma, \beta)$.
The BV anti-bracket on observables is determined by the natural $(-1)$-shifted symplectic structure on the shifted cotangent space---explicitly this utilizes the natural integration pairing between sections of $E$ and sections of $E^* \otimes \op{Dens}_M$.
The action functional is given by
\beqn
\int_M \beta (Q \gamma) ,
\eeqn
so that the equations of motion imply, in particular, that $\gamma$ is $Q$-closed.

Some examples of this constructions may be familiar. 
If $M$ is just a smooth manifold, then taking $\cE = \Omega^\bu(M)[1]$ to be the shifted smooth de Rham complex of $M$ results in the BV description of abelian topological $BF$ theory.
If $M = \Sigma$ is a Riemann surface then $\cE = \Omega^{0,\bu}(\Sigma, K^{\otimes r})$ recovers the spin $r$ bosonic ghost system used in string theory.
%If $M$ is a spin manifold then we can take $\cE = \Gamma(\cS)$, sections of the spinor bundle, and $Q = \slashed \del$ to be the Dirac operator.

Next, let $\cL$ be a local $L_\infty$ algebra on $M$ and suppose that the elliptic complex $\cE$ is a local $\cL$-module.
The coupling
\beqn\label{eqn:SLbackground}
S^{\cL}(\alpha;\beta,\gamma) = \int_M \beta (\alpha \cdot \gamma) ,
\eeqn
determines an element $S^{\cL} \in \op{Act}(\cL;\T^* [-1] \cE)$.
By the virtue that $\cE$ is a $\cL$-module this coupling automatically satisfies the equivariant classical master equation \eqref{eqn:equivariantmaster}.
Thus, any local $\cL$-module $\cE$ defines a classical $\cL$-background for the resulting free BV theory $\T^*[-1] \cE$.

Since $\cE$ is a free theory there is no obstruction to lifting $S^\cL$ to a quantum $\cL$-background.
There is, nevertheless, an obstruction to lifting this to an inner quantum $\cL$-background.
This obstruction admits an elegant description.

\begin{thm}[\cite{GwilliamThesis}]\label{thm:owen}
Let $(\cE,Q)$ be an elliptic complex which is equipped with a local $\cL$-action, where $\cL$ is some local Lie algebra.
Endow the BV theory $\T^*[-1] \cE$ with the structure of a classical $\cL$-background via $S^{\cL}$ in equation \eqref{eqn:SLbackground}.
The anomaly to the existence of a quantum inner $\cL$-background for $\T^*[-1] \cE$ is the trace of the action of $\cL$ on $\T^*[-1] \cE$.
\end{thm}

We specialize to when $M = X$ is a complex $n$-manifold.
The basic free holomorphic theory we consider is labeled by a holomorphic vector bundle $V$ on $X$.
The elliptic complex is the Dolbeault complex resolving the sheaf of holomorphic sections of $V$:
\beqn\label{eqn:dolbeaultelliptic} 
\cE_{V} = \Omega^{0,\bu}(X, V) .
\eeqn
The action functional $\int_X \beta \dbar \gamma$ returns the familiar free action of the $\beta\gamma$-ghost system used in string theory.
Notice that in dimension $n$, however, that the field $\beta$ is an element
\beqn
\beta \in \Omega^{n,\bu}(X,V^*)[n-1],
\eeqn
whose degree zero component is a field of Dolbeault type $(n,n-1)$.

The following will be the most important case for us.
We assume that $V \to X$ is a \textit{natural} holomorphic vector bundle; meaning a bundle built from the tangent bundle by taking duals, tensor products, direct sums, cohomological shifts etc..
This implies that the space of holomorphic sections of $V$ is equipped with the structure of a representation for the Lie algebra of holomorphic vector fields.
The action is by the Lie derivative.
This action extends in a natural way to a local $\cT_X$-module structure on the elliptic complex $\cE_V$ resolving the sheaf of holomorphic sections of $V$.

Summarizing, we have the following.

\begin{prop}
If $V \to X$ is a natural holomorphic vector bundle then the corresponding BV theory $\T^* [-1] \cE_V$ is equipped with a $\cT_X$-background defined by the coupling
\beqn\label{eqn:Tbackground}
S^{\cT}(\mu ; \beta,\gamma) = \int_X \beta (L_\mu \gamma) ,
\eeqn
where $L_\mu$ denotes the Lie derivative.
\end{prop}

In this context we can use \ref{thm:owen} to obtain a geometric description of the obstruction to this lifting to an inner $\cT_X$-background.

\subsection{Riemann--Roch theorem and anomalies}

Let $\cM$ denote the moduli stack of complex $n$-manifolds.
For the purposes of this note we will not need a precise model for this moduli space, but we will use the following features:
\begin{enumerate}
\item There is a universal complex manifold $\pi \colon \cX \to \cM$ whose fiber over an $n$-manifold $X \in \cM$ is $X$ itself.
\item Any natural bundle (one built from the holomorphic tangent bundle by taking duals, tensor products, direct sums etc.) defines a vector bundle on the universal $n$-manifold.
For example, there is a universal tangent bundle $\cT_{univ} \to \cX$ whose restriction to a complex $n$-manifold $X$ yields $\T_{X}$.
\item A model for the tangent complex of $\cM$ at a fixed complex $n$-manifold $X$ is the dg Lie algebra
\beqn
\mathbb{T}_X \cM \simeq \cT_X = \Omega^{0,\bu}(X, \T_X) .
\eeqn
\end{enumerate}

The last item is a formal version of the formalism of Kodaira and Spencer which says that the formal moduli problem for deforming complex structures is controlled by the dg Lie algebra $\cT_X$.
This means that a formal deformation of a complex $n$-manifold is a Beltrami differential
\beqn
\mu \in \Omega^{0,1}(X, \T_X)
\eeqn
which satisfies the Maurer--Cartan equation
\beqn\label{eqn:beltramimc}
\dbar \mu + \frac12 [\mu,\mu] = 0.
\eeqn
The deformed complex $n$-manifold has Dolbeault complex
\beqn
\left(\Omega^{0,\bu}(X) , \dbar + \mu\right) .
\eeqn
In other words, a formal deformation has the effect of modifying the $\dbar$ operator as $\dbar \rightsquigarrow \dbar + \mu$.
The condition that this squares to zero is precisely the Maurer--Cartan equation.

Now suppose that $V \to \cX$ is a universal holomorphic vector bundle.
We then obtain a family of elliptic complexes $\cE_{V}$ over $\cX$ whose fiber over a complex $n$-manifold $X \in \cM$ is the elliptic complex $\cE_{X,V}$ from \eqref{eqn:dolbeaultelliptic}.
To get a family over the moduli of complex structures $\cM$ we take the (derived) pushforward along $\pi \colon \cX \to \cM$ to get the family of elliptic complexes $\RR \pi_* \cE_{V}$ over $\cM$.
The determinant of this virtual bundle is a special case of the determinant, or anomaly, line bundle \cite{Freed, Quillen} associated to a family of $\dbar$-operators parameterized by a moduli of complex structures.

Now we consider the formal completion of $\cM$ at $X$.
By the expectation (3) above, the formal moduli space $\cM^{\;\Hat{}}_X$ is described by the dg Lie algebra $\cT_X$:
\beqn
\cM^{\;\Hat{}}_X \simeq \text{B} \cT_X .
\eeqn
The resulting family $\RR \pi_* \cE_V$ over the formal completion corresponds to a $\cT_X$-representation which is precisely encoded by the $\cT_X$-background \eqref{eqn:Tbackground}.

The consequence of theorem \ref{thm:owen} is that the one-loop anomaly for an inner $\cT_X$-background is the first Chern class of the virtual bundle $\RR \pi_* \cE_V$.

\begin{prop}
Let $V$ be a natural holomorphic vector bundle and consider the resulting $\cT_X$-background for the BV theory $\T^* [-1] \cE_{X,V}$.
When $X = \C^n$, the anomaly to lifting this to an inner $\cT_X$-background is given by the class
\beqn
\left. \op{Td} \cdot \ch(V)\right|_{2n+2} \in H^{2n+2} (BU(n)) \simeq H^{2n+1} (\lie{vect}(n)) \hookrightarrow H^1_{loc}(\cT(X)) .
\eeqn
\end{prop}

The consequence of this result in the complex one-dimensional case is hopefully familiar.
The ghost system for the bosonic string is the free theory based off of the complex $\cE_{\Sigma, \T_\Sigma [1]}$, where $\T_\Sigma[1]$ is the holomorphic tangent bundle shifted down by one.
The resulting anomaly in terms of universal Chern characters is
\beqn
\Td \cdot \ch(\T[1])|_4 = - \left(1 + \frac12 c_1 + \frac{1}{12} c_1^2\right)\left(1 + c_1 + \frac12 c_1^2\right)|_4 = - \frac{13}{12} c_1^2 .
\eeqn
The typical normalization for the generator of the Virasoro central charge is $\frac1{24} c_1^2$, implying that the central charge of this free system is $c = -26$ as expected.
To compare, the anomaly associated to the complex free boson CFT is $\frac1{12} c_1^2$; thus recovering the fact that the obtained by gauging $13$ complex copies of the free boson by holomorphic vector fields is anomaly free.
%We will comment more on this in section \ref{??}.

Let's move to the complex $n=2$ dimensional situation.
As an example, consider the case where $V$ is the trivial rank one bundle.
Then the universal Riemann--Roch theorem implies that the anomaly is represented by the local cocycle corresponding to the universal class 
\beqn
\Td|_6 = - \frac{1}{24} \ch_1 \ch_2 + \frac{1}{48} \ch_1^3 .
\eeqn
In other words $a^{hol}(V = triv) = -1/24$, $c^{hol}(V = triv) = 1/48$ so that
\beqn
- 4 \pi^2 \Theta_{V = triv} = -\frac{1}{24} \int_{\C^2} \op{tr} ( J \mu ) \op{tr}(\del J \mu \del J \mu) + \frac{1}{48} \int_{\C^2} \op{tr}(J \mu) \op{tr}(\del J \mu) \op{tr}(\del J \mu) .
\eeqn

We can generalize this example to an arbitrary natural line bundle which we will write as $V = K^\lambda$ for $\lambda$ a rational number, where $K$ is the canonical bundle on~$\C^2$.
We record the result here.

\begin{lem}\label{lem:grrfree}
There are the following expressions for the holomorphic central charges associated to the free holomorphic theory $\T^* [-1] \cE_{K^{\lambda}}$:
\beqn
a^{hol} (K^\lambda) = \frac{1}{24} r \quad \text{and} \quad c^{hol}(K^\lambda) = - \frac{1}{48} r^3 ,
\eeqn
where $r = 2 \lambda -1$.
\end{lem}
\begin{proof}
This follows from the expression in universal Chern classes
\beqn
\left. \Td \cdot \ch (K^\lambda)\right|_6 = \frac{1}{12} \left(\lambda - \frac12\right) \ch_1 \ch_2 - \frac16 \left(\lambda - \frac12\right)^3 \ch_1^3 ,
\eeqn
which is straightforward to verify.
\end{proof}

Observe that under $\lambda \mapsto 1- \lambda$ one has $r \mapsto - r$.
In particular, $a^{hol}$ and $c^{hol}$ satisfy 
\beqn
a^{hol}(K^\lambda) = - a^{hol} (K^{1-\lambda}), \quad c^{hol}(K^\lambda) = - c^{hol} (K^{1-\lambda}) .
\eeqn

\subsection{Computing the anomaly using the holomorphic gauge}

Recall that the free part of a holomorphic theory is determined by the $\dbar$ operator associated to some holomorphic vector bundle.
In \cite{BWhol} it is shown that any holomorphic theory admits a natural gauge where the gauge fixing operator is $Q^{GF} = \dbar^\star$.
The corresponding renormalization is well-behaved.
On flat space $\C^n$ the main results concerning renormalization for holomorphic theories can be summarized as follows.

\begin{thm}[\cite{BWhol}]
\label{thm:bw}
Suppose $(\cE,\omega,S = S_{free} + I)$ is a holomorphic theory on $\C^n$ and let $S[L]$ denote the effective action, at scale $L$, constructed using the gauge fixing operator $Q^{GF} = \dbar^\star$ whose corresponding propagator is $P_{\epsilon <L}$. 
\begin{enumerate}
\item To first-order in $\hbar$ the effective action is finite
\beqn
I[L] = \lim_{\epsilon \to 0} W(P_{\epsilon}^L , I) \mod \hbar^2
\eeqn
In particular, all one-loop counterterms are identically zero.
\item Suppose the free part of the action is of the form $S_{free} (\Phi) = \omega(\Phi, \dbar \Phi)$.
The anomaly to solving the quantum master equation to first-order in $\hbar$ can be expressed as the sum over wheels with exactly $(d+1)$-vertices. 
Explicitly, as a local functional this anomaly is
\beqn\label{eqn:generalanomaly}
\Theta_{1-loop} = \lim_{L \to 0} \lim_{\epsilon \to 0} \sum_{\Gamma \in {\rm Wheel}_{n+1}, e} W_\Gamma(P_{\epsilon < L}, K_\epsilon,I) .
\eeqn
where the sum is over all wheels with $(n+1)$-vertices and distinguished edges thereof.
Here $K_\epsilon$ is the heat kernel associated to the generalized Laplacian built using the gauge fixing operator $\dbar^\star$.
\end{enumerate}
\end{thm}

Let's return to the context of the previous section.
We consider the free BV theory~$\cE_{V}$ on $X = \C^n$, where $V$ is some natural vector bundle, in the $\cT$-background~$S^{\cT}(\xi;\beta,\gamma)$ from~\eqref{eqn:Tbackground}.
In this case we can use the trivialization $V \simeq \C^n \times V_0$ to identify the $\dbar$-operator for $V$ with $\dbar \otimes \id$ where $\dbar$ is simply the ordinary $\dbar$-operator on $\C^n$.
Correspondingly, we choose the gauge fixing operator is $\dbar^\star \otimes \id$ where $\dbar^\star$ is Hodge adjoint of $\dbar$ constructed using the flat hermitian metric on $\C^n$. 
Explicitly, the heat kernel is given by
\beqn
K_L(z,w) = \frac{1}{(2\pi i L)^n} e^{-|z-w|^2/4L} \otimes \id_V \in \cE_{V} \Hat{\otimes} \cE_V [1] 
\eeqn
where $\id_V \in V_0 \otimes V_0^*$ represents the identity morphism $V_0 \to V_0$.
The propagator is 
\beqn
P_{\epsilon<L} (z,w) = \int_{T = \epsilon}^L \d T \, (\dbar^\star \otimes \id) \, K_T(z,w) .
\eeqn

Since $\T^*[-1]\cE_{V}$ is a free BV theory, item (1) of theorem \ref{thm:bw} implies that 
\beqn
S^{\cL} [L] \define \lim_{\epsilon \to 0} W(P_{\epsilon}^L , I) \mod \hbar^2 
\eeqn
is a one-loop quantum $\cT$-background.
Further, item (2) implies that the obstruction to lifting this to a one-loop inner $\cT$-background is
\beqn
\Theta = \lim_{L \to 0} \lim_{\epsilon \to 0} \sum_{\Gamma \in {\rm Wheel}_{n+1}, e} W_\Gamma(P_{\epsilon < L}, K_\epsilon,I) ,
\eeqn
which automatically determines a degree one cocycle in $\cloc^\bu(\cT(\C^n))$.

We specialize to complex dimension $n=2$.
In this case, the anomaly $\Theta$ is expressed as a sum over wheels with three vertices.
Since the effective action is $GL(2,\C)$-invariant it is straightforward to see that based on symmetry the anomaly local cocycle $\Theta$ is cohomologous to a linear combination of the cocycles $\Theta_a^{hol}$ and $\Theta_c^{hol}$ from proposition \ref{prop:holac}.
That is
\beqn
- 4 \pi^2 \Theta_V = a^{hol}(V) \Theta_{a^{hol}} + c^{hol}(V) \Theta_{c^{hol}}
\eeqn
for some coefficients $a^{hol}(V),c^{hol}(V)$.
Instead of computing these coefficients explicitly using the effective action, which appears to be quite cumbersome, we will refer to the Riemann--Roch theorem as outlined in the previous subsection to fix them.
%In the appendix we have worked out the coefficient using the effective action in the case~$n=1$.
%\brian{fix this citation}

\subsection{Riemann--Roch theorem for mixed anomalies}
\label{sec:mixed2}

So far we have only discussed anomalies for the action of holomorphic vector fields on free theories within the BV formalism.
These anomalies are avatars of gravitational, or conformal, anomalies as portrayed in the physics literature.
For non-free theories, like gauge theories, there are also possible \textit{internal} anomalies, like gauge anomalies, which must be trivial in order for quantization to exist.

In this section we consider anomalies present when holomorphic gauge theory is put in a $\cT$-background.
In addition to the holomorphic gravitational anomalies we have been focused on thus far, there will also be pure gauge and mixed anomalies which must be trivialized.
To make the discussion as parallel as possible with the case of pure diffeomorphism anomalies, we consider the case that both diffeomorphisms and gauge symmetries are considered as background fields.
From this we will be able to deduce mixed anomalies for interacting gauge theories.

Let $G$ be a Lie group.
Let $\cM_G$ denote the moduli of pairs $(X,P)$ of a complex $n$-manifold $X$ together with a principal holomorphic $G$-bundle $P$ on $X$.

As in the case of complex $n$-manifolds, we will not develop a precise model for this moduli space, but we will use the following expected features:
\begin{enumerate}
\item There is a universal moduli $\pi \colon \cX_G \to \cM_G$ whose fiber over a pair $(X,P) \in \cM_G$ is the principal bundle $P \to X$ itself.
\item Any natural bundle (one built from the holomorphic tangent bundle by taking duals, tensor products, direct sums etc.) together with a representation of $G$ defines a vector bundle on the universal $n$-manifold.
%For example, there is a universal tangent bundle $\cT_{univ} \to \cX$ whose restriction to a complex $n$-manifold $X$ yields $\T_{X}$.
\item A model for the tangent complex of $\cM$ at the pair $(X,P_{triv})$, where $P_{triv}$ is the trivial $G$-bundle, is the semi-direct product dg Lie algebra
\beqn
\mathbb{T}_X \cM_G \simeq = \Omega^{0,\bu}(X, \T_X) \ltimes \Omega^{0,\bu}(X) \otimes \lie{g} .
\eeqn
This dg Lie algebra is a resolution of the sheaf of holomorphic sections of the sheaf of Lie algebras $\cT^{hol}_X \ltimes \cO^{hol} \otimes \lie{g}$.
\end{enumerate}

The last point says that a formal deformation of a pair $(X,P_{triv})$ is a Beltrami differential
$\mu \in \Omega^{0,1}(X, \T_X)$, satisfying \eqref{eqn:beltramimc}, as before together with a $(0,1)$ connection 
\beqn
A \in \Omega^{0,1}(X) \otimes \lie{g}
\eeqn
which together satisfy the Maurer--Cartan equation
\beqn
(\dbar + \mu) A + \frac12 [A,A] = 0.
\eeqn
The deformed complex $n$-manifold has Dolbeault operator $\dbar + \mu$ and the deformed principal bundle has Dolbeault operator $\dbar + A$.

\begin{prop}
\label{prop:mixedRR}
Let $V$ be a natural holomorphic vector bundle equipped with a $G$-action and consider the resulting $\cT_X \ltimes \Omega^{0,\bu}_X \otimes \lie{g}$-background for the BV theory $\T^* [-1] \cE_{X,V}$.
When $X = \C^n$, the anomaly to lifting this to an inner $\cT_X \ltimes \Omega^{0,\bu}_X \otimes \lie{g}$-background is given by the class
\beqn
\left. \op{Td} \cdot \ch_{GL(n) \times G} (V)\right|_{2n+2} \in H^{2n+2} (BGL(n) \times BG) \hookrightarrow H^1_{loc}(\cT(X) \ltimes \Omega^{0,\bu}(X) \otimes \lie{g}) ,
\eeqn
where $\ch_{GL(n) \times G} (V)$ is the $GL(n) \times G$-character of the representation $V$.
\end{prop}

Let's consider a simple example when the complex dimension $n=1$ and $G = GL_1(\C) = \C^\times$.
Suppose that $V$ is the trivial super vector bundle $W[\eps] = W \oplus \Pi W$ where $W$ is some natural vector bundle on a curve.
In this case, the pure holomorphic gravitational anomaly for the BV theory $\T^* [-1] \cE_{\C,W[\eps]}$ vanishes.
So, the anomaly corresponds to an element in
\beqn
H^2(BGL_1) \otimes H^2(B G) \subset H^4(BGL_1 \times B G), 
\eeqn
which turns out to be proportional to $c_1 \til c_1$ where $c_1$ represents the `gravitational' first Chern class and $\til c_1$ is the background gauge $G=GL_1(\C)$ class.
Explicitly, as a local functional this class is represented by
\beqn
\int_\C J \mu \op{Tr}(\del A) 
\eeqn
where $\mu$ is a background holomorphic vector field and $A$ is a background gauge field.
This anomaly was found independently in \cite{GRWthf} where it is interpreted as the famous Calabi--Yau anomaly for the two-dimensional $B$-model.
Indeed, the holomorphic theory we are considering here is a holomorphic twist of the two-dimensional $\cN=2$ chiral multiplet valued in $W$.
The $GL_1(\C) \times GL_1(\C)$ represents, at the group level, the dilations and $R$-symmetry which remains inside of the two-dimensional $\cN=2$ superconformal algebra after twisting.

\subsection{Holomorphic gauge theories}

In this section we have only discussed anomalies for \textit{background} holomorphic diffeomorphism symmetry or \textit{background} holomorphic gauge symmetry where it suffices to just consider the underlying free theory.
For non-free theories, like gauge theories, there are also possible \textit{internal} anomalies, like gauge anomalies, which must be trivial in order for quantization to exist.
In this section we explain how holomorphic diffeomorphism symmetry mixes with possible internal gauge symmetries.

The idea is the following.
We start with a holomorphic gauge theory which has a background symmetry by holomorphic vector fields.
Before discussing the issue of quantizing this background symmetry, we first need to know that the gauge theory has no internal anomaly.
Once this internal anomaly is trivialized, we can then treat the underlying free theory as a theory with a background symmetry by $\cT \ltimes \Omega^{0,\bu} \otimes \lie{g}$.
We then treat the problem just as we did before with two caveats.
First, since we have already trivialized the pure gauge anomaly, it suffices to only look at anomalies which depend non-trivially on a holomorphic vector field.
Second, the terms in the anomaly which have a nontrivial dependence on the gauge field represent the anomaly to having a quantum $\cT$-background.
Once these are trivialized, we can then contemplate the anomaly to having a quantum inner $\cT$-background; this is represented by the terms in the anomaly which are independent of the gauge field.

Let's consider an explicit example.
Let $\lie{g}$ be a Lie algebra of a compact Lie group~$G$.
The basic holomorphic gauge theory we consider is holomorphic $BF$ theory which can be defined on any complex manifold $X$.
Globally, holomorphic $BF$ theory is the shifted cotangent bundle to the moduli stack of holomorphic $G$-bundles on $X$.
Within the perturbative BV formalism the fields consist of pairs 
\begin{align*}
A & \in \Omega^{0,\bu}(X, \lie{g})[1] \\
B & \in \Omega^{n,\bu}(X, \lie{g}^*)[n-2]
\end{align*}
and the action is of standard $BF$-type: $S_{BF} = \int_X B F_A$.
This theory enjoys a classical backround $\cT$-symmetry defined by
\beqn
S^\cT = \int_X B (L_\mu A) .
\eeqn

Holomorphic $BF$ theory suffers from an internal gauge anomaly.
As usual, we focus just on the case $X = \C^n$.
Using theorem \ref{thm:bw} one can show the following. 

\begin{prop}
\label{prop:puregauge}
Holomorphic $BF$ theory on $\C^n$ is anomalous at one-loop in perturbation theory.
The local functional representing this anomaly is
\beqn
- \frac{1}{(n+1)!}\int_{\C^n} \op{Tr}_{ad} \left(A (\del A)^n\right) ,
\eeqn
where the trace denotes the trace is in the adjoint representation of $\lie{g}$ on itself.
\end{prop}

%The local functional representing the pure gauge anomaly is automatically a cocycle in the local complex 
%\beqn
%\cloc^\bu \left( \Omega^{0,\bu}(\C^n, \lie{g}) \right)
%\eeqn
%of cohomological degree one.
%In \cite{GWkm} an injective map of cochain complexes
%\beqn
%\Sym^{n+1}(\lie{g})^{G} \hookrightarrow \cloc^\bu \left( \Omega^{0,\bu}(\C^n, \lie{g}) \right) [1] .
%\eeqn
%The anomaly of holomorphic $BF$ theory is represented by the invariant polynomial
%\beqn
%X \in \lie{g} \mapsto \op{Tr}_{ad} (X^{n+1}) ,
%\eeqn
%which we will denote by $\op{ch}_{n+1}(\T_{B \lie{g}})$.
For example, in complex dimension one $n=1$, the holomorphic $BF$ theory associated to any semi-simple Lie algebra is anomalous.
In complex dimension two $n=2$ the anomaly for holomorphic $BF$ theory vanishes for complex semi-simple Lie algebras.
%(As is the case for twists of four-dimensional supersymmetric QCD, see section~\ref{sec:qcd}.)

We now turn to thinking of the underlying free theory as a theory with a background symmetry by $\cT \ltimes \Omega^{0,\bu} \otimes \lie{g}$.
Such anomalies to quantization will be labeled by certain elements in
\beqn
H^{2n+2} (BGL(n) \times BG) 
\eeqn
as in the last section.
The thing that distinguishes this discussion from the last section is that we are assuming that the component in $H^{2n+2}(BG)$, which we just summarized in the previous proposition, vanishes
Explicitly, this is the component $-\ch_{n+1}(\T_{B\lie{g}}) = - \ch_{n+1}^G(\lie{g}_{ad})$.

The following is a consequence of proposition \ref{prop:mixedRR}.
\begin{prop}\label{prop:Tgauge}
Suppose $\ch_{n+1}(\T_{B\lie{g}}) = 0$.
The anomaly to lifting $S^{\cT}$ to a quantum $\cT$-background is represented by the class
\begin{multline}
\label{eqn:Tgauge}
- \left. \op{Td} \cdot \op{ch}^G(\lie{g}^{ad}) \right|_{2n+2}^\star = - \sum_{j=1}^{n} \Td_{j} \ch_{n+1-j} (\T_{B \lie{g}}[1]) \\
\in \oplus_{j} H^{2j}(BGL(n)) \otimes H^{2n+2-2j}(BG) \subset H^{2n+2}(BGL(n) \times BG) 
\end{multline}
where the $\star$ indicates that we do not include the component $\Td_{2n+2}$.
\end{prop}

The reason we have not included the term $\Td_{2n+2}$ is that this term represents the anomaly to having an inner $\cT$-action, which is only well-defined once the above class vanishes.
It is very rare that pure holomorphic $BF$ theory admits a quantum $\cT$-background; the above anomaly conditions are very restrictive.
We obtain many examples however, once we include matter.
We will detail an explicit example motivated by four-dimensional supersymmetry in the next section, but for now we state the general result.

Suppose that $V$ is a $G$-equivariant natural holmorphic vector bundle on $X$.
Examples of such bundles include $K_X^{\otimes \lambda} \otimes U$ where $U$ is a $G$-representation.
In particular, on $X = \C^n$, we can consider its $GL(n) \times G$-character
\beqn
\op{ch}^{GL(n) \times G} (V) \in H^\bu(GL(n) \times G) .
\eeqn

We consider coupling to holomorphic $BF$ theory for the Lie group $G$ the following holomorphic matter system valued in $V$.
The fields, in the BV formalism, are
\beqn
\gamma \in \Omega^{0,\bu}(X, V)
\eeqn
together with the conjugate field $\beta \in \Omega^{n,\bu}(X, V)[n-1]$.
The total action functional is 
\beqn
\int_X B F_A + \int_X \beta \dbar_A \gamma
\eeqn
where $\dbar_A = \dbar + A$.
We refer to this as holomorphic BF theory coupled to the $\beta \gamma$ system with values in $V$.
The classical $\cT$-background is 
\beqn
S^{\cT} = S^\cT = \int_X B (L_\mu A)  + \int_X \beta (L_\mu \gamma) .
\eeqn
The proof of the proposition below combines the results above.

\begin{prop}
\label{prop:coefficients}
Holomorphic BF theory coupled to the $\beta \gamma$ system with values in $V$ on $\C^n$ is anomalous at one-loop in perturbation theory.
The invariant representing this anomaly is
\beqn\label{eqn:anomalycoefficients}
\ch^G_{n+1} (V) - \ch^G_{n+1}(\lie{g}^{ad}) \in H^{2n+2} (BG) .
\eeqn

Suppose the invariant \eqref{eqn:anomalycoefficients} vanishes.
The anomaly to lifting $S^{\cT}$ to a quantum $\cT$-background is represented by the class
\beqn
\label{eqn:Tcoefficients}
\left. \op{Td} \cdot \left(\op{ch}^{GL(n) \times G}(V) - \op{ch}^G(\lie{g}^{ad}) \right) \right|_{2n+2}^\star 
\in H^{2n+2}(BGL(n) \times BG) 
\eeqn
where the $\star$ indicates that we do not include the component in $H^{2n+2}(BGL(n))$.

Suppose the invariant \eqref{eqn:Tcoefficients} vanishes.
Then the $\cT$-central charge of the system is 
\beqn
\left. \op{Td} \cdot \left(\op{ch}^{GL(n) \times G}(V) - \op{ch}^G(\lie{g}^{ad}) \right) \right|_{2n+2} \in H^{2n+2} (BGL(n)) .
\eeqn
\end{prop}

\subsection{Compactification}

Suppose we have a holomorphic theory defined on $\C^n \times X$ where $X$ is some compact complex manifold of dimension $m$.
Then, we obtain a holomorphic theory on $\C^n$ via compactification along $\C^n \times X \to \C^n$.
We turn to the problem of characterizing the $\cT_{\C^n}$-anomaly polynomial from the anomaly polynomial in $n+m$ dimensions.

Recall that the space of anomalies in $(n+m)$-dimensions is
\beqn
H^1_{loc}(\cT_{\C^{n+m}}) = H^{2n+2m+1}(\lie{vect}(n+m)) = H^{2n+2m+2}(BU(n+m)) .
\eeqn
Let $X$ be a compact $m$-dimensional manifold, then we have the composition
\beqn
\begin{tikzcd}
H^{2n+2m+2} (BU(n+m)) \ar[d,equal] \ar[r] & H^{2n+2m+2}(BU(n) \times BU(m)) \ar[r] & H^{2n+2}(BU(n)) \ar[d,equal] \\ 
H^{1}_{loc}(\cT_{\C^{n+m}}) \ar[rr,dotted,"\int_X"] & & H^{1}_{loc}(\cT_{\C^{n}}) .
\end{tikzcd}
\eeqn
The first map on the top line is induced from the block diagonal $U(n) \times U(m) \to U(n+m)$ and the second map restricts along the tangent classifier $X \to BU(m)$ and then integrates along $X$.
We denote this composition by $\int_X$.

\begin{prop}
If $\Theta$ is the anomaly to a quantum $\cT$-background on $\C^n \times X$ then $\int_X \Theta$ is the anomaly to a quantum $\cT$-background on $\C^n$.
\end{prop}

As an example, consider the basic higher dimensional $\beta \gamma$ system on $\C \times \Sigma$ where $\Sigma$ is a Riemann surface. 
The fields are 
\begin{align*}
\gamma & \in \Omega^{0,\bu}(\C \times \Sigma) \\
\beta & \in \Omega^{2,\bu}(\C \times \Sigma)[1]  .
\end{align*}
The compactification along $\Sigma$ is simply the ordinary $\beta \gamma$ system on $\C$ with values in the graded vector space $H^\bu(\Sigma, \cO)$\footnote{More accurately, this is the two-dimensional bosonic $\beta\gamma$ system with values in $H^0(\Sigma, \cO)$ coupled to the two-dimensional fermionic $bc$ system with values in $H^1(\Sigma, \cO)$.}, whose fields are
\begin{align*}
\gamma_{2d}& \in \Omega^{0,\bu}(\C) \otimes H^\bu(\Sigma, \cO) \\
\beta_{2d} & \in \Omega^{1,\bu}(\C) \otimes H^{\bu}(\Sigma, K) [1]  .
\end{align*}
In particular, the the $\cT_{\C}$ anomaly polynomial is simply
\beqn
\chi^{hol}(\Sigma,\cO) \cdot \frac{1}{12} c_1^2 \in H^{4}(BU(1)) ,
\eeqn
where $\chi^{hol}(\Sigma, \cO)$ is the holomorphic Euler characteristic.

On the other hand, we have seen that the universal anomaly on $\C^2$ for this higher dimensional $\beta\gamma$ system is simply
\beqn
\op{Td}|_{6} = - \frac{1}{24} \op{ch}_1 \op{ch}_2 + \frac{1}{48} \op{ch}_1^3 \in H^6(BU(2)) .
\eeqn
Now, $\int_\Sigma \op{ch} = 1 + \op{ch}_1 + \int_\Sigma c_1(\Sigma)$ where $c_1(\Sigma)$ is the first Chern class of the tangent bundle of $\Sigma$, and $\int_\Sigma \op{ch}_2 = \frac12 \op{ch}_1^2$.
Since $\chi^{hol}(\Sigma, \cO) = \frac12 \int_\Sigma c_1(\Sigma)$ we have
\beqn
\int_\Sigma \op{Td}|_6 = - \frac{1}{24} \cdot \frac{1}{2} \int_\Sigma c_1(\Sigma) \op{ch}_1^2 + \frac{1}{48} \cdot 3 \cdot \int_{\Sigma} c_1(\Sigma) \ch_1^2 = \frac{1}{12} \chi^{hol}(\Sigma, \cO) \ch_1^2 \in H^4(BU(1)) . 
\eeqn

\section{Relationship to the $a,c$ anomalies}
\label{s:rsymmetry}

Most of the work up until now held in arbitrary complex dimension.
Now we specialize to complex dimension two, and specifically to holomorphic quantum field theories obtained from twisting four-dimensional superconformal theories.

As reviewed in the introduction, to a four-dimensional conformal theory one associates two invariants.
These are the coefficients $c$ of (Weyl)$^2$ and $a$ of (Euler) in the conformal anomaly $\<T_\mu^\mu\>$.

A four-dimensional \textit{super}conformal theory is, in part, a theory which has a symmetry by the four-dimensional superconformal algebra. 
There are variants of this algebra depending on the amount of supersymmetry, but in this section we will only need the minimal $\cN=1$ supersymmetric version.
In particular, being superconformal means that the theory has a symmetry by the so-called $R$-symmetry group $U(1)$ which commutes with supersymmetry (and conformal transformations).
Crucially, for this section, in such a four-dimensional superconformal theory, the coefficients $a,c$ are related to certain coefficients in the $R$-symmetry anomaly which we recall momentarily.

This paper concerns holomorphic quantum field theories.
To any four-dimensional supersymmetric theory there is a twist (unique up to $R$-symmetry, if it exists) which renders the theory holomorphic \cite{CostelloHol,CostelloYangian,ESW}. 
A natural question is: what structure does the holomorphic twist of a \textit{superconformal} field theory posess?
The answer to this question is one of the main results of \cite{SWsuco}.

\begin{thm}[\cite{SWsuco}]
Suppose $\sfT$ is the holomorphic twist of four-dimensional $\cN=1$ superconformal theory of Yang--Mills type.\footnote{By this, we mean a theory consisting of some number of chiral and vector multiplets with both gauge interactions and possibly a superpotential.
In particular this is not just \textit{pure} gauge theory as there are possibly matter fields as well.}
Then $\sfT$ admits a classical $\cT$-background.
\end{thm}

In other words, the holomorphic twist of a $\cN=1$ superconformal theory admits a symmetry by holomorphic vector fields.
In this sense, we argue in \cite{SWsuco} that the Lie algebra of holomorphic vector fields is an enhanced twisted superconformal algebra.

The main result of this section shows how the coefficients $a,c$ of the supersymmetric theory relate to the $\cT$-equivariant anomaly polynomial of the holomorphically twisted theory.

\begin{thm}
\label{thm:ac}
Consider the holomorphic twist of a four-dimensional superconformal theory on~$\R^4$ of Yang--Mills type and let $a,c$ be the gravitational anomaly coefficients of this supersymmetric theory.
The anomaly polynomial for $\cT$-equivariant quantization of the holomorphic twist is 
\beqn\label{eqn:acholTheta}
\Theta = \frac23 (a-c) \ch_1 \ch_2 + \frac19 \left(c - \frac{5}{3} a\right) \ch_1^3 .
\eeqn
\end{thm}

In other words, there is the following relation involving the holomorphic anomalies $a^{hol},c^{hol}$ (see equation \eqref{eqn:achol}):
\begin{align*}
a^{hol} &= \frac23 (a-c) \\
c^{hol} &= \frac{1}{9} \left(c - \frac53 a\right) 
\end{align*}
Alternatively, the expression for $a,c$ in terms of $a^{hol},c^{hol}$ is
\beqn\label{eqn:fromhol}
\begin{aligned}
a & = - \frac{9}{4} \left(a^{hol} + 6 c^{hol} \right) \\
c & = - \frac{3}{4} \left(5 a^{hol} + 18 c^{hol} \right) .
\end{aligned}
\eeqn

\subsection{$R$-symmetry and twisting data}

A four-dimensional superconformal theory is, in particular, a theory with $R$-symmetry.
The minimal amount of $R$-symmetry is with respect to the group $U(1)$, which we denote by $U(1)_R$ when speaking of $R$-symmetry.

Let $R$ be the generator of the $U(1)_R$-symmetry in the original supersymmetric theory.
We will utilize the following well-known relationship between the $a,c$ coefficients of the anomaly polynomial and the $R$-symmetry, see \cite{Anselmi:1997ys,Anselmi:1997am}:
\beqn\label{eqn:ac1}
\begin{aligned}
a & = \frac{3}{32} \left(3 \op{Tr} R^3 - \op{Tr} R \right) \\
c & = \frac{1}{32} \left(9 \op{Tr} R^3 - 5 \op{Tr} R\right) .
\end{aligned}
\eeqn
For standard multiplets, the $R$-charges appearing in the above formulas are those of the chiral fermions in the theory.
In terms of the anomaly polynomial for mixed gravitational and $R$-symmetry (in this case, $U(1)_R$) the coefficients $(\op{Tr} R^3), (\op{Tr} R)$ appear as
\beqn\label{eqn:R-symmetryanomaly}
\Theta^{\cN=1} = \frac{1}{3!} \left( (\op{Tr} R^3) \til c_1^3 - \frac14 (\op{Tr} R) \til c_1 p_1\right)
\eeqn
where $\til c_1$ is the generator of $H^2(BU(1)_R)$ and $p_1 \in H^4 (BSO(4))$ is the universal first Pontryagin class, which we have normalized with a factor of $-1/4$ to fit with physics conventions.
Our goal is to relate these coefficients to objects in the holomorphic twist.

Generally speaking, the operation of `twisting' a supersymmetric theory involves the following steps.
Throughout, we fix a supersymmetry algebra together with an $R$-symmetry group $G_R$.
\begin{enumerate}
\item Find a supercharge $Q$, which is square zero.
That is, $Q$ is an odd element of the supersymmetry algebra which satisfies $[Q,Q] = 0$.
\item Find a `grading homomorphism' $\alpha\colon U(1) \to G_R$, where $G_R$ is the $R$-symmetry group. 
The supercharge $Q$ must be weight one with respect to this homomorphism.
\item Find a `twisting homomorphism'.
This amounts to choosing a Lie group $G$, together with homomorphisms $\iota \colon G \to Spin(n)$ and $\phi \colon G \to G_R$ with the property that $Q$ is $G$-invariant with respect to the induced homomorphism $\iota \times \phi \colon G \to Spin(n) \times G_R$.
\end{enumerate}

The four-dimensional $\cN=1$ supersymmetry algebra has an essentially unique supercharge which is square-zero \cite{EStwist}.
Any other supercharge is in the same $R$-symmetry orbit.
We fix such a supercharge $Q^{hol}$.
This supercharge is \textit{holomorphic}, in the sense that the rank of the map $[Q^{hol},-]$ is two (half the dimension of spacetime).
For more details on twisting we refer to \cite{CostelloHol,ESW,SWchar}.

\begin{prop}
\label{prop:twist}
Consider a free supersymmetric theory on $\R^4$ built from chiral and vector multiplets with $R$-symmetry generator $R = (r_i)$ where $r_i$ is the $R$-charge of the $i$th multiplet.
Our convention is that for chiral (respectively, vector) multiplets, the scalar (respectively gauge) field has $R$-charge $r_i+1$ (respectively, $r_i-1$).

Then, there exists twisting data such that the holomorphic twist of of this theory is equivalent to the cotangent theory (in the BV sense) to
\beqn
\oplus_i \Omega^{0,\bu}(\C^2, K^{(r_i+1)/2}) .
\eeqn
\end{prop}

As an example, consider a single free vector multiplet with its standard $R$-charge $r = 1$ (so that the gauge field in the multiplet has vanishing $R$-charge $R(A) = 0$).
This lemma implies that the holomorphic twist is the cotangent theory to $\Omega^{0,\bu}(\C^2,K) = \Omega^{2,\bu}(\C^2)$ which consists of fields 
\beqn
(B,A) \in \Omega^{2,\bu}(\C^2) \oplus \Omega^{0,\bu}(\C^2) [1],
\eeqn
thus recovering the free limit of holomorphic $BF$ theory.

As another example, the $R$-charge for a single chiral multiplet which saturates the BPS bound for the superconformal algebra is $r = -\frac13$ (so that the scalar in the multiplet has $R$-charge $R(\phi)=\frac23$).
The holomorphic twist of this theory is the theory whose underlying space of BV fields is 
\beqn
(\gamma, \beta) \in \Omega^{0,\bu}(\C^2, K^{1/3}) \oplus \Omega^{0,\bu}(\C^2, K^{2/3})[1] ,
\eeqn
and action is $\int \beta \dbar \gamma$.
This is the two-dimensional $\beta\gamma$ system twisted by $K^{1/3}$.

\begin{proof}[Proof of theorem \ref{thm:ac}]
To compute the $\cT$-anomaly it suffices to consider the underlying free theory, hence we can appeal to proposition \ref{prop:twist}.

From \eqref{eqn:ac1} we have the following relation between the $R$-symmetry generator and the gravitational anomaly coefficients:
\beqn
\op{Tr} R = 16 (a-c) , \quad \op{Tr} R^3 = \frac{16}{3} \left(\frac53 a - c\right) .
\eeqn

From proposition \ref{prop:twist} we know that we can express the holomorphic twist as a $\beta\gamma$-system with coefficients in a sum of line bundles determined by $R$. 
Using Lemma \ref{lem:grrfree} we see that
\beqn
a^{hol} = \frac{1}{24} 16 (a-c) = \frac{2}{3} (a-c) 
\eeqn
and
\beqn
c^{hol} = - \frac{1}{48} \cdot \frac{16}{3} \left(\frac53 a - c\right) = \frac19 \left(c - \frac53 a\right)
\eeqn
as desired.

\end{proof}

There is an alternative way of deducing these holomorphic anomaly coefficients from the anomaly polynomial of the supersymmetric theory.
Indeed, we can start with the $R$-symmetry anomaly polynomial of the untwisted theory in equation \eqref{eqn:R-symmetryanomaly} and proceed with the following two steps: first, break $SO(4)$ to $MU(2)$ and second, twist by setting $\til c_1 = -\frac12 \op{ch}_1$.
The first step has the effect of setting $p_1 \mapsto 2 \op{ch}_2$ (this is the induced homomorphism $H^4(BSO(4)) \to H^4(BU(2))$ and the second step follows from the standard twisting homomorphism for four-dimensional supersymmetry, see \cite{ESW}.

%
%For $\cN=4$ supersymmetric Yang--Mills theory we know that
%\beqn
%(a,c)_{\cN=4} = \left(\frac14, \frac14\right) .
%\eeqn

\subsection{Anomalies for basic supersymmetric theories}

We unpack the anomaly polynomial for basic four-dimensional supersymmetric gauge theories.
The results, including anomalies for theories with more supersymmetry, are summarized in table \ref{tab:basic}.

\begin{table}
\centering
\begin{tabular}{ |c|c|c|c|c|c|}
 \hline
  $\cN$ & multiplet & $a$ & $c$ & $a^{hol}$ & $c^{hol}$ \\
 \hline
 \hline
 $\cN=1$   & vector   & $\frac{3}{16}$ & $\frac{1}{8}$ & $\frac{1}{24}$ & $-\frac{1}{48}$ \\
 & chiral &   $\frac{1}{48}$  & $\frac{1}{24}$   & $-\frac{1}{72}$ & $\frac{1}{1296}$ \\
 \hline
 $\cN=2$ &  vector  & $\frac{5}{24}$   & $\frac{1}{6}$ & $\frac{1}{36}$ & $-\frac{13}{648}$ \\
  & hyper & $\frac{1}{24}$ &  $\frac{1}{12}$ & $-\frac{1}{36}$ & $\frac{1}{648}$ \\ 
  \hline
  $\cN=4$  & vector & $\frac14$ & $\frac14$ & $0$ & $-\frac{1}{54}$ \\
 \hline
\end{tabular}
\caption{A list of $a,c$ and $a^{hol},c^{hol}$ for basic supersymmetric gauge theories and their twists.}
\label{tab:basic}
\end{table}

\subsubsection{$\cN=1$ gauge theory}

In $\cN=1$ gauge theory there are vector multiplets and chiral multiplets.
We have described the twist of the superconformal vector and chiral multiplet after the statement of proposition \ref{prop:twist}.
To compute the coefficients $a^{hol},c^{hol}$ we could refer to the formula \eqref{eqn:acholTheta}.
Instead, we will compute the coefficients directly from the Riemann--Roch theorem.

In fact, we have already done this.
The anomaly polynomial for the chiral multiplet is obtained by plugging $\lambda = \frac13$ into the coefficients of lemma \ref{lem:grrfree}.
Similarly, the anomaly polynomial for a single vector multiplet is obtained as the negative of the coefficients in proposition \ref{lem:grrfree} when $\lambda = 0$.
The sign arises because this is a gauge theory.

Translating back to the physical coefficients as in equation \eqref{eqn:fromhol} we obtain the expected values of $a,c$ \cite{Duff2}.

\subsubsection{$\cN=2$ gauge theory}

We focus on two multiplets relevant for four-dimensional $\cN=2$ supersymmetric gauge theory.
The first is the vector multiplet, which is defined from the data of a Lie group $G$ and a principal bundle, which we take to be trivial.
The second is the hypermultiplet, which depends on the choice of a representation $R$ of $G$.

In the BV formalism, the holomorphic twist of pure gauge theory is holomorphic $BF$ theory for the holomorphic local Lie algebra whose underlying holomorphic vector bundle is 
\beqn
\lie{g} \otimes \wedge^\bu K^{1/3} = \lie{g} \oplus \Pi \lie{g} \otimes K^{1/3} .
\eeqn
The Lie algebra structure is determined by the Lie bracket on $\lie{g}$.
Note that
\beqn\label{eqn:wedgeK13}
\ch\left(\lie{g}[1] \otimes \wedge^\bu K^{1/3} \right) = \left(- \frac13 \ch_1 + \frac{1}{18} \ch_1^2 - \frac{1}{162} \ch_1^3 \right) \dim G .
\eeqn
Thus
\beqn
\Td \ch\left(\lie{g}[1] \otimes \wedge^\bu K^{1/3} \right)|_6 = \left(\frac{1}{36} \ch_1 \ch_2 - \frac{13}{648} \ch_1^3 \right) \dim G .
\eeqn
(Alternatively, note that as $\cN=1$ multiplets a $\cN=2$ vector is the same as a $\cN=1$ vector plus a $\cN=1$ chiral).

The hypermultiplet is the cotangent theory to 
\beqn
R \otimes \wedge^\bu K^{1/3} \otimes K^{1/3} = R \otimes K^{1/3} \oplus \Pi R \otimes K^{2/3},
\eeqn
From this we see that the anomaly polynomial for the hypermultiplet valued in $R$ is
\beqn
\left(-\frac{1}{36} \ch_1 \ch_2 + \frac{1}{648} \ch_1^3 \right) \dim R .
\eeqn

Translating back to the physical coefficients as in equation \eqref{eqn:fromhol} we obtain the expected values of $a,c$ \cite{Duff2}.

\subsubsection{$\cN=4$ gauge theory}
Consider four-dimensional $\cN=4$ supersymmetric gauge theory for gauge group~$G$ near the trivial $G$-bundle.
In the BV formalism, the holomorphic vector bundle underlying the holomorphic twist of this theory is
\beqn
\lie{g}[1] \otimes \wedge^\bu\left(K^{1/3} \otimes \C^3\right) . 
\eeqn
Notice that in the cohomology of $BGL(2)$ we have the formal expression
\beqn
\ch(K^{\lambda/3}) = 1 - \frac{\lambda}{3} \ch_1 + \frac{\lambda^2}{2 \cdot 3^2} \ch_1^2 - \frac{\lambda^3}{3! \cdot 3^3} \ch_1^3 . 
\eeqn
where $\lambda$ is any integer.
From this we compute
\begin{align*}
\ch \wedge^\bu\left(K^{1/3} \otimes \C^3\right) & = 1 - 3 \ch (K^{1/3}) + 3 \ch (K^{2/3}) - \ch(K) \\
& = \frac{1}{27} \ch_1^3 .
\end{align*}

Thus, we see directly that for $\cN=4$ supersymmetric Yang--Mills theory we have $a^{hol} = 0$ and 
\beqn
c^{hol} = -\frac{1}{2 \cdot 27} \dim G = -\frac{\dim G}{54} .
\eeqn

Translating back to the physical coefficients as in equation \eqref{eqn:fromhol} we obtain the expected values of $a,c$ \cite{Duff2}.

\section{Anomaly matching in holomorphic QCD}
\label{s:qcd}

In this section we consider anomalies in the holomorphic twist of four-dimensional $\cN=1$ supersymmetric quantum chromodynamics (QCD).

\subsection{The holomorphic theory}

Four-dimensional $\cN=1$ supersymmetric Yang--Mills theory depends on a choice of compact Lie group $G$, with complexified Lie algebra $\fg$, and a representation $\bV$ that comprises the supersymmetric matter.
%The field content of the theory on a spin Riemannian four-manifold $M$~is
%\begin{itemize}
%\item a  that consists of a gauge field $A \in \Omega^1(M) \otimes \fg $ together with fermions $\lambda \in \Gamma(M,S) \otimes \fg$, and
%\item a {\it matter multiplet} (often called the chiral multiplet), depending on a choice of a complex representation $V$ of $\fg$, that consists of a scalar $\phi \in C^\infty(M) \otimes V$ together with fermions $\psi \in \Gamma(M,S) \otimes V$. 
%Additionally, one can turn on a {\it superpotential} $W \in \cO(V)^\fg$, a $G$-invariant polynomial on~$V$. 
%\end{itemize}

Typically, in supersymmetric QCD one takes $G = SU(N_c)$ and 
\beqn
\bV = (\C^{N_c})^{\oplus N_f} \oplus (\br \C^{N_c})^{\oplus N_f}
\eeqn
where $\C^{N_c}$ (respectively, $\br \C^{N_c}$) denotes the fundamental (respectively, anti-fundamental)  $SU(N_c)$ representation.
One says that $N_c$ is the number of `colors' and $N_f$ is the number of `flavors'.

In the $\cN=1$ supersymmetry algebra there is (up to equivalence) a unique twisting supercharge.

\begin{thm}[\cite{CostelloYangian,ESW,SWchar}]
The twist of four-dimensional $\cN=1$ supersymmetric Yang--Mills theory on $\R^4$ associated to the pair $(G,\bV)$ is equivalent to holomorphic BF theory on $\C^2$ with gauge group $G_\C$ coupled to the holomorphic $\beta\gamma$ system on~$\C^2$ valued in~$\bV$.

If the supersymmetric theory is equipped with a superpotential $W$, then there is a holomorphic superpotential in the twisted theory, see below.
\end{thm}

Following this theorem, we arrive at the holomorphic version of supersymmetric QCD obtained from the previous theorem as the twist of supersymmetric QCD with $N_c$ colors and $N_f$ flavors.
Explicitly, this the holomorphic theory whose fields are:
\begin{itemize}
\item the gauge fields are those of holomorphic $BF$ theory valued in $\lie{sl}(N_c)$, whose fields are
\[
A \in \Omega^{0,\bu}(\C^2 , \lie{sl}(N_c))[1] \quad\text{and}\quad B \in \Omega^{2,\bu}(\C^2, \lie{sl}(N_c)^*);
\]
\item the ``quarks'' are fields of a holomorphic $\beta\gamma$ system valued in 
\[
V = (\C^{N_c})^{N_f},
\] 
where $\C^{N_c}$ is the fundamental $SU(N_c)$ representation whose fields are
\[
\gamma \in \Omega^{0,\bu}(\C^2, V) \quad\text{and}\quad \beta^\dag \in \Omega^{2,\bu}(\C^2 , V^*)[1];
\]
\item the ``antiquarks'' are fields of a holomorphic $\beta\gamma$ system valued in the dual representation 
\[
V^* \cong (\br \C^{N_c})^{N_f},
\]
where $\br \C^{N_c}$ is the anti-fundamental $SU(N_c)$ representation whose fields are
\[
\gamma^\dag \in \Omega^{0,\bu}(\C^2 , V^*) \quad\text{and}\quad \beta \in \Omega^{2,\bu}(\C^2 , V) [1].\footnote{The superscript $\dag$ does not refer to adjoint here.}
\]
\end{itemize}

The action functional is
\[
S_{QCD} = \int B F_A + \int \beta^\dag \dbar_A \gamma + \int \beta \dbar_A \gamma^\dag,
\]
where $\dbar_A = \dbar + A$ is the covariant $\dbar$-operator.

If, additionally, the supersymmetric theory is equipped with a superpotential $W$, then there is a term in the holomorphic action
\beqn
\int \d^2 z \, W(\gamma) = \int \d^2 z \, \left(\frac12 W''(\gamma^0) \gamma^{0,1} \gamma^{0,1} + W'(\gamma^0) \gamma^{0,2}\right) .
\eeqn
Notice that this depends on the choice of a holomorphic volume form and does not define a theory with a cohomological $\Z$-grading.
Both of these drawbacks stem from the fact that the superpotential breaks $R$-symmetry in the original supersymmetric theory.
%We will denote this theory by $\cT(N,F)$.

In the above description we have described the fields on flat space.
Even on flat space, a $\cT$-background will depend on which bundle the quark fields transform under.
For example, for a single superconformal chiral multiplet, the field $\gamma$ is twisted by the line bundle $K^{1/3}_{\C^2}$.

For holomorphic QCD we will do something more general.
We suppose that $\gamma$ (and $\gamma^\dag$) transform as sections of $K^{\otimes \lambda}_{\C^2} \otimes V$ (and $K^{\otimes \lambda}_{\C^2} \otimes V^*$) where $\lambda$ is some rational number.
From section \ref{s:rsymmetry}, recall that this corresponds to the twist of a chiral multiplet with $r = 2 \lambda - 1$.

The $\cT$-background is 
\beqn\label{eqn:Tr}
I^{\cT,r} = \int B L_\mu A + \int \beta^\dagger L_\mu \gamma + \int \beta L_\mu \gamma^\dagger .
\eeqn
With respect to the standard coordinates $\mu = \mu_i (z) \del_{z_i}$ and this coupling reads
\beqn
I^{\cT,r} = \int B \mu_i \del_{z_i} A + (1-r) \int \beta^\dagger \mu_i \del_{z_i} \gamma + (1-r) \int \beta \mu_i \del_{z_i} \gamma^\dag + r \int \del_{z_i} \mu_i (\beta^\dagger \gamma + \beta \gamma^\dagger) .
\eeqn

%\begin{lem}
%Let $\til \gamma = \gamma (\d^2 z)^{\otimes r}$ and $\til \beta = \frac{\del}{\del (\d^2 z)} (\beta) (\d^2 z)^{1-r}$ and similarly for $\til \gamma^\dagger,\til \beta^\dagger$.
%Then
%\beqn
%I^{\cT,r}(A,B,\beta, \gamma, \beta^\dagger,\gamma^\dagger) = I^{\cT} (A,B,\til \beta, \til\gamma, \til\beta^\dagger,\til\gamma^\dagger) .
%\eeqn
%\end{lem}
%\begin{proof}
%This follows from the formula for the Lie derivative
%\beqn
%L_\mu (\gamma (\d^2 z)^r) = \left((1-r) \mu_i \del_{z_i} \gamma + r \del_{z_i} \mu_i \gamma\right) (\d^2 z)^r .
%\eeqn
%\end{proof}
%
%This lemma implies that the generalized coupling $I^{\cT,r}$ is the standard coupling $I^{\cT}$ for the theory whose fields are `twisted' by certain powers of the canonical bundle; namely $\til \gamma \in \Omega^{0,\bu} (\C^2, K^{\otimes r})$, etc..

\subsection{There are no gauge anomalies}

Consider holomorphic QCD with gauge group $SU(N_c)$ and~$N_f$ flavors valued in the fundamental representation.
First we consider the internal anomalies.
For pure $\lie{sl}(N_c)$ holomorphic gauge theory is no pure gauge anomaly as $\op{Tr}_{adj}(X^3) = 0$ for any $X \in \lie{sl}(N_c)$, see proposition \ref{prop:puregauge}.
The mixed gauge-matter anomalies vanish since the matter comprises an equal number of fundamental and anti-fundamental representations.

\subsection{Anomaly to a quantum $\cT$-background}
\label{sec:TbackgroundQCD}

Next we consider the anomaly to having background of holomorphic vector fields $\cT$ using the generalized coupling $I^{\cT,r}$ from equation \eqref{eqn:Tr}.
We find an exact agreement with the anomaly cancelation for the $R$-symmetry current in supersymmetric QCD.

\begin{prop}
Holomorphic QCD admits a quantum $\cT$-background with classical coupling~$I^{\cT,r}$ if and only if
\beqn
r = - \frac{N_c}{N_f} ,
\eeqn
equivalently, $\lambda = \frac12 (1-\frac{N_c}{N_f})$.
\end{prop}
\begin{proof}
This is a consequence of proposition \ref{prop:coefficients}.
The graded holomorphic vector bundle underlying holomorphic QCD is 
\beqn
V_{QCD} = \lie{g}[1] \oplus K^{\otimes \lambda} \otimes V \oplus K^{\otimes \lambda} \otimes V^* .
\eeqn
In the expansion of \eqref{eqn:Tgauge} we see that since the gauge group is simple, the only the terms in anomaly live in
\beqn
H^2(BGL(2)) \otimes H^4(B G) ,
\eeqn
where here $G = SU(N_c)$.
Explicitly, this is the term
\beqn\label{eqn:td0td1}
\op{Td}_0 \ch_3^{GL(2) \times G}(V_{QCD}) + \op{Td}_1 \ch_2^{G}(V_{QCD}) .
\eeqn

Using the form of $V_{QCD}$ above, we see that the first term in equation \eqref{eqn:td0td1} is
\beqn
2 F \ch_2(K^{\otimes \lambda}) = - 2\lambda F \ch_1 \op{ch}^G_2(fun) .
\eeqn
We have used the relation $\op{ch}_1(K^{\otimes \lambda}) = - \lambda \ch_1$ where, as usual, $\op{ch}_1$ is the standard generator of $H^2(BGL(2))$. 
The factor of two comes from the fact that we have quarks and anti-quarks.

Likewise, the second term in \eqref{eqn:td0td1} is
\beqn
\frac{1}{2} \ch_1 \left(- \ch_1 \op{ch}^G_2(adj) + \cdot 2 F \cdot \ch_2^G(fun)\right) .
\eeqn

In total we see that the anomaly to a quantum $\cT$-background is
\beqn
\ch_1 \left(-\frac12 \ch_2^{G}(adj) + 2 \left(-\lambda + \frac12\right) F \ch_2^{G}(fun) \right).
\eeqn
For any $X \in \lie{sl}(N_c)$ we have the relation $\op{Tr}_{adj}(X^2) = 2 N_c \op{Tr}_{fun}(X^2)$.
In other words, for $G = SU(N_c)$, we have $\op{ch}_2^G(adj) = 2 N_c \op{ch}_2^G(fun)$. 
Using this relation, the result follows.
%%that the wheels with two legs labeled by the gauge field $A$ and one leg labeled by a vector field $\mu$ vanish.
%Indeed, the total contribution of the anomaly is
%\beqn
%-\int \op{tr}(J \mu) \op{Tr}(\del A \del A) + 2 N_f (1-r) \int \op{tr}(J \mu) \op{tr}(\del A \del A) = 0 .
%\eeqn
%The first term comes from the contribution of the gauge propagators; using the Riemann-Roch theorem it corresponds to the universal class $- c_1 \cdot \ch_2(\T_{B G})$, where $G  = SU(N_c)$ here.
%The second term comes from the contribution of the matter propagators and corresponds to the class $c_1 \cdot \ch_2(V_Q + V_Q^*)$.
%Anomaly cancellation then requires 
%\beqn
%- N_c (N_c^2 - 1) + 2 N_f \frac{N_c^2 - 1}{2 N_c} N_c (1-r) = 0 ,
%\eeqn
%which completes the proof.
\end{proof}

\subsection{$a,c$ for QCD}

We are now in a position to compute the values of $a^{hol},c^{hol}$ for holomorphic QCD.
As always, this only depends on the underlying holomorphic vector bundle of holomorphic QCD
\beqn
V_{QCD} = \ul{\lie{g}}[1] \oplus K^{\otimes \lambda} \otimes V \oplus K^{\otimes \lambda} \otimes V^* ,
\eeqn
where, as we found in the last section, one has $\lambda = \frac12 (r+1) = \frac12 \left(1 - \frac{N}{F}\right)$ in order to guarantee anomaly cancellation for the $\cT$-background.

According to this decomposition, the $a^{hol}$ coefficient decomposes as $a^{hol}_{QCD} = a^{hol}_{gauge} + a^{hol}_{quark} + a^{hol}_{anti-quark}$, and similarly for $c^{hol}$.
Notice that $a^{hol}_{quark} = a^{hol}_{anti-quark}$ and $c^{hol}_{quark} = c^{hol}_{anti-quark}$ as they transform identically under the action of holomorphic vector fields.
We read these coefficients directly.

First
$a^{hol}_{gauge} = \frac{1}{24} \dim \lie{g} = \frac{1}{24} \left(N_c^2 -1\right)$,
as read off from table \ref{tab:basic}.
From lemma \ref{lem:grrfree} we see
\beqn
a^{hol}_{quark} = \frac{1}{24} \left(- \frac{N_c}{N_f} \right) \dim V = - \frac{1}{24} N_c^2 .
\eeqn
Thus in total we see that
\beqn\label{eqn:aqcdhol}
a^{hol}_{QCD} = - \frac{1}{24} \left(N_c^2 + 1\right)
\eeqn
Interestingly, this is independent of the number of flavors $N_f$.

Similarly, from table \ref{tab:basic}, $c^{hol}_{gauge} = -\frac{1}{48} \dim \lie{g} = - \frac{1}{48} (N^2 - 1)$,
and from \ref{lem:grrfree} we see
\beqn
c^{hol}_{quark} = - \frac{1}{48} \left(- \frac{N_c}{N_f}\right)^3 \dim V = \frac{1}{48} \frac{N_c^4}{N_f^2} .
\eeqn
Thus
\beqn\label{eqn:cqcdhol}
c^{hol}_{QCD} = \frac{1}{48N_f^2}\left(2N_c^4 - N_c^2 N_f^2+N_f^2\right).
\eeqn

Plugging into equation \eqref{eqn:fromhol} we recover the expected physical values for these coefficients
\begin{align*}
a_{QCD} & = -\frac{3}{16 N_f^2} \left(3N_c^4-2N_c^2 N_f^2+N_f^2\right) \\
c_{QCD} & = - \frac{1}{16 N_f^2} \left(9 N_c^4 - 7 N_c^2 N_f^2 + 2 N_f^2 \right) .
\end{align*}

\subsection{Magnetic anomalies}

In \cite{Seiberg}, Seiberg discovered a duality in four-dimensional supersymmetric gauge theory that posits an equivalence between $\cN=1$ supersymmetric QCD for gauge group $SU(N_c)$, called the `electric' theory, and another $\cN=1$ supersymmetric QCD, called the `magnetic' theory.

The magnetic theory has gauge group $SU(N_f - N_c)$ and the matter consists of $N_f$ fundamental chiral multiplets together with $N_f^2$ `mesons' which transform trivially under the gauge group.
There is also a superpotential which we will not need so do not comment any further on.
Therefore, at the level of the holomorphic twist, the bundle underlying the magnetic dual to holomorphic QCD is thus of the form
\beqn
\til V_{QCD} = \ul{\til{\lie{g}}}[1] \oplus K^{\otimes \til \lambda} \otimes \til V \oplus K^{\otimes \til \lambda} \otimes \til V^* \oplus K^{\otimes \lambda_M} \otimes \op{End}(\C^F) ,
\eeqn
where $\til{\lie{g}} = \lie{sl}(N_f-N_c)$ and the matter is labeled by the representation $\til V = (\C^{N_f - N_c})^{N_f}$.

Since the holomorphic mesons transform trivially under $\til{\lie{g}}$ they do not contribute to any possible gauge anomalies.
In particular, we can compute the anomaly to a quantum $\cT$-background just as we did in section \ref{sec:TbackgroundQCD}.
We find that a quantum $\cT$-background exists if and only if
\beqn
\til \lambda = \frac{N_c}{2 N_f},
\eeqn
equivalently $\til r = - \frac{N_f - N_c}{N_f}$.

From this, we can proceed as above to calculate the $a,c$ anomalies for the magnetic version of QCD.
As a function of $r_M = 2\lambda_M-1$, we find
\begin{align*}
\til a^{hol}_{QCD} & = \frac{1}{24} \left(-1 + 2 N_f N_c - N_c^2 + N_f^2(r_M - 1) \right) \\
\til c^{hol}_{QCD} & = \frac{1}{48 N_f^2}\left(N_f^2 + N_f^4 - 6 N_f^3 N_c + 11N_f^2 N_c^2-8 N_f N_c^3+2N_c^4-r_M^3 N_f^4\right) .
\end{align*}

\begin{prop}
The $a,c$ anomalies match for the holomorphic twists of Seiberg dual theories
\beqn
a_{QCD} = \til a_{QCD}, \quad c_{QCD} = \til c_{QCD}
\eeqn
if and only if $r_M = 1 - \frac{2N_c}{N_f}$. 
\end{prop}
\begin{proof}
It suffices to prove the statement for the holomorphic coefficients $a^{hol},c^{hol}$. 
Immediate computation shows that $a_{QCD}^{hol}$ from equation \eqref{eqn:aqcdhol} matches with $\til a^{hol}_{QCD}$ if and only if $r_M = 1 - \frac{2N_c}{N_f}$.
Plugging this value of $r_M$ into $\til c_{QCD}^{hol}$ we find a match with equation \eqref{eqn:cqcdhol}.
\end{proof}

%\subsection{Holomorphic flavor symmetry}

%\newcommand{\sfa}{\mathsf{a}} 
%
%In addition to a classical $\cT$-symmetry, holomorphic QCD enjoys a holomorphic `flavor' symmetry \cite{GWkm}.
%In other words, exhibits a natural $\lie{sl}(N_f) \otimes \Omega^{0,\bu}(\C^2)$-background defined by the coupling
%\beqn
%I^{flav} = \int \beta^\dag(\sfa \cdot \gamma) + \int \beta (\sfa \cdot \gamma^\dag) ,
%\eeqn 
%where $\sfa \in \lie{sl}(N_f) \otimes \Omega^{0,\bu}(\C^2)$ denotes a background holomorphic gauge field.

\printbibliography

\end{document}